\documentclass[letterpaper,10pt,journal]{IEEEtran}
\IEEEoverridecommandlockouts
\usepackage{url}
\usepackage{amssymb}
\usepackage{amsfonts}
\usepackage{amsthm}
\usepackage{amsmath,amssymb,verbatim}
\usepackage{graphicx}
\usepackage{subcaption}
\usepackage{float}
\usepackage{epsfig}
\usepackage{threeparttable}
\usepackage{algorithm}
\usepackage{algorithmic}
\usepackage{color} 
\usepackage{changepage}
\usepackage{comment}
\usepackage{multirow}
\usepackage{colortbl}
\usepackage{ulem}
\usepackage{color}
\usepackage{lipsum}
\usepackage{cite}
\usepackage{tabularx}
\usepackage{dsfont}
\usepackage{ulem}

\floatname{algorithm}{Algorithm}

\newcommand{\blue}{\textcolor{black}}

\newtheorem{theorem}{\it Theorem}

\newtheorem{corollary}{\it Corollary}

\newcommand{\bE}{\mathds{E}}
\newcommand{\mK}{\mathcal{K}}
\newcommand{\mN}{\mathcal{N}}

\DeclareMathOperator*{\argmax}{argmax}

\begin{document}

\title{Waiting but not Aging: Optimizing Information Freshness Under the Pull Model}

\author{
Fengjiao Li,~\IEEEmembership{Student Member,~IEEE,}
Yu Sang,
Zhongdong Liu,~\IEEEmembership{Student Member,~IEEE,}
Bin Li,~\IEEEmembership{Senior Member,~IEEE,}
Huasen Wu,~\IEEEmembership{Member,~IEEE,}
and Bo Ji,~\IEEEmembership{Senior Member,~IEEE} 
\thanks{
This work was supported in part by the NSF under Grants CCF-1657162, CNS-1651947, and CNS-1717108. A preliminary version of this work was presented at IEEE GLOBECOM 2017 \cite{sang17}.

Fengjiao Li (fengjiaoli@vt.edu), Zhongdong Liu (zhongdong@vt.edu), and Bo Ji (boji@vt.edu) are with the Department of Computer Science, Virginia Tech, Blacksburg, VA. Bin Li (binli@uri.edu) is with the Department of Electrical, Computer and Biomedical Engineering,
University of Rhode Island, Kingston, Rhode Island. 
Bo Ji is the corresponding author.
}
}
\maketitle

\begin{abstract}
The Age-of-Information is an important metric for investigating the timeliness performance in information-update systems. 
In this paper, we study the AoI minimization problem under a new Pull model with replication schemes, where a user proactively sends a replicated request to multiple servers to ``pull" the information of interest. Interestingly, we find that under this new Pull model, replication schemes capture a novel tradeoff between 
different values of the AoI across the servers (due to the random updating processes) and different response times across the servers, which can be exploited to minimize the expected AoI at the user's side. 
Specifically, assuming Poisson updating process for the servers and exponentially distributed response time, we derive a closed-form formula for computing the expected AoI and obtain the optimal number of responses to wait for to minimize the expected AoI. Then, we extend our analysis to the setting where the user aims to maximize the AoI-based utility, which represents the user's satisfaction level with respect to freshness of the received information. Furthermore, we consider a more realistic scenario where the user has no prior knowledge of the system. In this case, we reformulate the utility maximization problem as a stochastic Multi-Armed Bandit problem with side observations and leverage a special linear structure of side observations to design learning algorithms with improved performance guarantees. Finally, we conduct extensive simulations to elucidate our theoretical results and compare the performance of different algorithms. Our findings reveal that under the Pull model, waiting does not necessarily lead to aging; waiting for more than one response can often significantly reduce the AoI and improve the AoI-based utility in most scenarios.
\end{abstract}

\section{Introduction}\label{sec:intro}
The last decades have witnessed the prevalence of smart devices and significant advances in ubiquitous computing and the Internet of things.
This trend is forecast to continue in the years to come \cite{vni2019}. 
The development of this trend has spawned a plethora of real-time services that require timely information/status updates. 
One practically important example of such services is vehicular networks and intelligent transportation systems~\cite{kaul11secon,kaul11globecom}, where accurate status information (position, speed, acceleration, tire pressure, etc.) of a vehicle needs to be shared with other nearby vehicles and road-side facilities in a timely manner in order to avoid collisions and ensure substantially improved road safety.
More such examples include sensor networks for environmental/health monitoring~\cite{ko10,corke10}, wireless channel feedback~\cite{costa15}, news feeds, weather updates, online social networks, flight aggregators (e.g., Google Flights), and stock quote services. 

For systems providing such real-time services, those commonly used performance metrics, such as throughput and delay, 
exhibit significant limitations in measuring the system performance \cite{kaul12infocom}. 
Instead, \emph{the timeliness of information updates becomes a major concern}. To that end, a new metric called the \emph{Age-of-Information (AoI)} 
has been proposed as an important metric for studying the timeliness performance \cite{kaul11secon}.
The AoI is defined as the time elapsed since the most recent update occurred (see Eq.~\eqref{eq:aoi} for a formal definition).
Using this new AoI metric, the work of \cite{kaul12infocom} 
employs a simple system model to analyze and optimize the timeliness performance of an information-update system.
This seminal work has recently aroused dramatic interests from the research community and has inspired a series
of interesting studies on AoI analysis and optimization (see \cite{aoisurvey,sun2020book}  and references therein).

While all prior studies consider a \emph{Push} model, concerning about when and how to 
``push" (i.e., generate and transmit) the updated information to the user, in this paper we introduce a new 
\emph{Pull} model, under which a user sends requests to the servers to proactively ``pull" the information of interest.  
This Pull model is more relevant for many important applications where the user's interest is in the freshness 
of information at the point when the user requests it rather than in continuously monitoring the freshness of information. 
One application of the Pull model is in the real-time stock quote service, where a customer (i.e., user) submits a query to multiple stock quote providers (i.e., servers) sharing common information sources and each provider responds with its most up-to-date information. Other applications include flight aggregators and real estate listings apps.

\emph{To the best of our knowledge, however, none of the existing work on the timeliness optimization has considered 
such a Pull model. In stark contrast, we focus on the Pull model and propose to employ request replication 
to minimize the AoI or to maximize the AoI-based utility at the user's side.} Although a similar Pull model is considered for data synchronization 
in \cite{bright04,bright06}, the problems are quite different and request replication is not exploited. 
Note that the concept of replication is not new and has been extensively studied for various applications 
(e.g., cloud computing and datacenters~\cite{gardner15,ananthanarayanan12}, 
storage clouds~\cite{li16}, parallel computing~\cite{wang14,wang15}, 
and databases \cite{pacitti99,pereira10}).
\emph{However, for the AoI minimization problem under the Pull model, replication schemes exhibit a 
unique property and capture a novel tradeoff between different levels of information freshness and different 
response times across the servers. This tradeoff reveals the power of waiting for more than one response 
and can be exploited to optimize information freshness at the user's side.
}

Next, we explain the above key tradeoff through a comparison with cloud computing systems.
It has been observed that in a cloud or a datacenter, the processing time of a same job can be highly 
variable on different servers \cite{ananthanarayanan12}.
Due to this important fact, replicating a job on multiple servers and waiting for the first finished copy can
help reduce the latency \cite{ananthanarayanan12,gardner15}. Apparently, in such a system it is \emph{not} 
beneficial to wait for more copies of the job to finish, as all the copies would give the same outcome. 
By contrast, in the information-update system we consider, although the servers may possess the same 
type of information (weather forecast, stock prices, etc.), they could have different versions of the 
information with different levels of freshness due to the random updating processes.
In fact, the first response may come from a server with stale information; waiting for more than one 
response has the potential of receiving fresher information and thus helps reduce the AoI. 
\emph{Hence, it is no longer the best to stop waiting after receiving the first response (as in the other aforementioned applications).}
On the other hand, waiting for too many responses will lead to a longer total waiting time, and thus, it also incurs a larger AoI at the user's side.
\emph{Therefore, it is challenging to determine the optimal number of responses to wait for in order to minimize the expected AoI (or to maximize the AoI-based utility) at the user's side.}
The problem is further exacerbated by the fact that the updating rate and the mean response time, which are important to making such decisions, are typically unknown to the user a priori. 


We summarize our key contributions as follows. 
\begin{itemize}
\item To the best of our knowledge, this work, for the first time, introduces the Pull 
model for studying the timeliness optimization problem and proposes to employ request replication to reduce the AoI.

\item Assuming Poisson updating process at the servers and exponentially distributed response time, we derive a closed-form formula for computing the expected AoI and obtain the optimal number of responses to wait for to minimize the expected AoI. We also discuss some extensions to account for more general replication schemes and different types of response time distributions.

\item We further consider scenarios where the user aims to maximize the utility, which is an exponential function of the negative AoI. The utility represents the user's satisfaction level with respect to freshness of the received information.
We derive a set of similar theoretical results for the utility maximization problem.

\item Moreover, we consider a more realistic scenario where the user has no prior knowledge of the system parameters such as the updating rate the and the mean response time.
In this case, we formulate the utility maximization problem as a stochastic Multi-Armed Bandit (MAB) problem with side observations. The side observations lead to the feedback graph with a special linear structure, which can be leveraged to design learning algorithms with improved regret upper bounds.

\item Finally, we conduct extensive simulations to elucidate our theoretical results. We also investigate the impact of the system parameters  on the achieved gain. 
Our findings reveal that under the Pull model, waiting does not necessarily lead to aging; waiting for more than one response can often significantly reduce the AoI and improve the AoI-based utility in most scenarios.
In the case of unknown system parameters, our simulation results show that algorithms exploiting the special linear feedback graph outperform the classic learning algorithms.
\end{itemize}

The remainder of this paper is organized as follows. 
We first discuss related work in Section~\ref{sec:related} and then describe our new Pull model in Section~\ref{sec:model}.
In Section~\ref{sec:aoi}, we analyze the expected AoI under replication schemes and obtain the optimal number of responses for minimizing the expected AoI. In Section~\ref{sec:utility}, we consider the utility maximization problem in the settings where the updating rate and the mean response time are known and unknown, respectively. Section~\ref{sec:sim} presents the simulation results, and we conclude the paper in Section~\ref{sec:conclusion}.

\section{Related Work}\label{sec:related}
Since the seminal work on AoI \cite{kaul11secon}, there has been a large body of work focusing on AoI analysis and optimization in a wide variety of settings and applications (see \cite{aoisurvey,sun2020book} for surveys). However, almost all prior work considers the Push model, in contrast to the Pull model we consider in this paper.

A series of work (e.g., \cite{kaul12infocom,kaul2012status,yates2019age,kam16tit,kam2016controlling,kam2018age,pappas15,costa16,bedewy19tit}) has been focused on analyzing the AoI performance of various queueing models. 
In \cite{kaul12infocom}, the authors analyze the expected AoI in M/M/$1$, M/D/$1$, and D/M/$1$ systems under the First-Come-First-Served (FCFS) policy.
A follow-up work in \cite{kam16tit} extends the analysis to M/M/$2$ and M/M/$\infty$ models. 
The expected AoI is also characterized for M/M/$1$ Last-Come-First-Served (LCFS) model, with and without preemption, for single-source and multi-source systems~\cite{kaul2012status,yates2019age}. 
Furthermore, controlling the AoI through packet deadlines is studied in \cite{kam2016controlling,kam2018age}; the effect of the packet management (e.g., prioritizing new arrivals and discarding old packets) on the AoI is considered in \cite{pappas15,costa16}.
In \cite{bedewy19tit}, the authors show that the preemptive Last-Generated-First-Served (LGFS) policy achieves the optimal (or near-optimal) AoI performance in a multi-server queueing system.

There have also been lots of recent efforts denoted to the design and analysis of AoI-oriented scheduling algorithms in various network settings (e.g., \cite{kadota16,kadota2019scheduling,talak2018opt,talak2018optimizing,lu2018age,joo18ton,talak2017minimizing,bedewy2019age,yates2018age}).
In \cite{kadota16}, the authors aim to minimize the weighted sum AoI of the clients in a broadcast wireless network with unreliable channels. A similar problem with throughput constraints is considered in a follow-up study~\cite{kadota2019scheduling}. In \cite{talak2018opt,talak2018optimizing,joo18ton}, the authors consider AoI-optimal scheduling problems in ad hoc wireless networks under interference constraints. Considering a similar network setting, the authors of \cite{lu2018age} aim to design AoI-aware algorithms for scheduling real-time traffic with hard deadlines.
Recently, the study on AoI has also been pushed towards more challenging settings with multi-hop flows~\cite{talak2017minimizing,bedewy2019age,yates2018age}.

We want to point out that the preliminary version of our paper \cite{sang17} is the first work that employs the Pull model and replication schemes to study the AoI at the user's side. Since then, the idea of replication has also been adopted for studying the AoI under different models (see, e.g., \cite{zhong2017status,zhong2018minimizing,bedewy19tit}). 
Recently, the authors of \cite{yin2019only} also aim to minimize the AoI from the users' perspective by considering multiple users.
Note that outside the AoI area, similar Pull models have been investigated (e.g., for data synchronization~\cite{bright04,bright06}) since decades ago. However, the problems they study are very different, and request replication is not exploited.

Besides the linear AoI considered in the above work, there are several studies that investigate more general functions of the AoI (e.g., \cite{sun17tit,kosta2017age,klugel2019aoi}). 
Such functions are often used to model utility/penalty, which represents the user's satisfaction/dissatisfaction level with respect to freshness of the received information.
In this extended journal version (see Section~\ref{sec:utility}), we also consider the AoI-based utility, which is an exponential function of the negative AoI. However, the model and the problem we consider are quite different from those in the existing work. Furthermore, we study the scenario where the system parameters are unknown and cast the AoI-based utility maximization problem as an online learning problem based on the stochastic MAB formulation.

Although variants of MAB formulations have recently been considered for AoI minimization problems (see, e.g., \cite{sun2019closed,hsu2018age,jiang2019timely}), they consider Markovian MAB, where the state of each arm evolves in a Markovian fashion and the reward drawn at each time is a function of the current state of the selected arm. This is very different from the stochastic MAB we consider, in terms of model, algorithm design, and regret analysis. Moreover, we consider a new Pull model and exploit a special linear structure of the feedback graph to design learning algorithms with improved regret upper bounds.

\section{System Model}\label{sec:model}
We consider an information-update system where a user pulls time-sensitive information from $n$ servers. 
These $n$ servers are connected to a common information source and update their data \emph{asynchronously}. 
We call such a model the \emph{Pull} model (see Fig.~\ref{fig:model}). Let $\mN \triangleq \{1, 2, \dots, n\}$ be the set of indices of the servers, and let $i \in \mN$ be the server index.
We assume that the information updates at the source for each server follow a Poisson process with rate $\lambda>0$ (where the update rate can model system limitations or resource budgets). The updating processes are asynchronous and are assumed to be independent
and identically distributed (\emph{i.i.d.}) across the servers. We also assume that there is no transmission delay from the source to the servers. That is, servers instantaneously receive their updates once the updates are generated at the source.
This implies that the inter-update time (i.e., the time duration between two successive updates) at each server 
follows an exponential distribution with mean $1/\lambda$. 
Let $u_i(t)$ denote the time when the most recent update at server $i$ occurs, and let $\Delta_i(t)$ denote the 
AoI at server $i$, which is defined as the time elapsed since the most recent update at this server:
\begin{equation}
\label{eq:aoi}
\Delta_i(t) \triangleq t - u_i(t). 
\end{equation}
Therefore, the AoI at a server drops to zero if an update occurs at this server; otherwise, the AoI increases linearly as time goes by until the next update occurs. Fig.~\ref{fig:update} provides an illustration of the AoI evolution at server $i$.

In this work, we consider the $(n,k)$ replication scheme, under which the user sends the replicated copies of the 
request to all $n$ servers and waits for the first $k$ responses. Let $R_i$ denote the response time for server $i$.
Note that each server may have a different response time, which is the time elapsed since the request is sent out 
by the user until the user receives the response from this server. We assume that the time for the requests to reach 
the servers is negligible compared to the time for the user to download the data from the servers. Hence, the response 
time can be interpreted as the downloading time. Let $s$ denote the downloading start time, which is the same for 
all the servers, and let $f_i$ denote the downloading finish time for server $i$. Then, the response time for server $i$ 
is $R_i = f_i - s$. We assume that the response time is exponentially distributed with mean $1/\nu$ and is \emph{i.i.d.}  
across the servers.
\blue{Note that the model we consider above is simple, but it suffices to capture the key aspects and novelty of the problem we study.}

\begin{figure}[!t]
\centering
\includegraphics[width=0.4\textwidth]{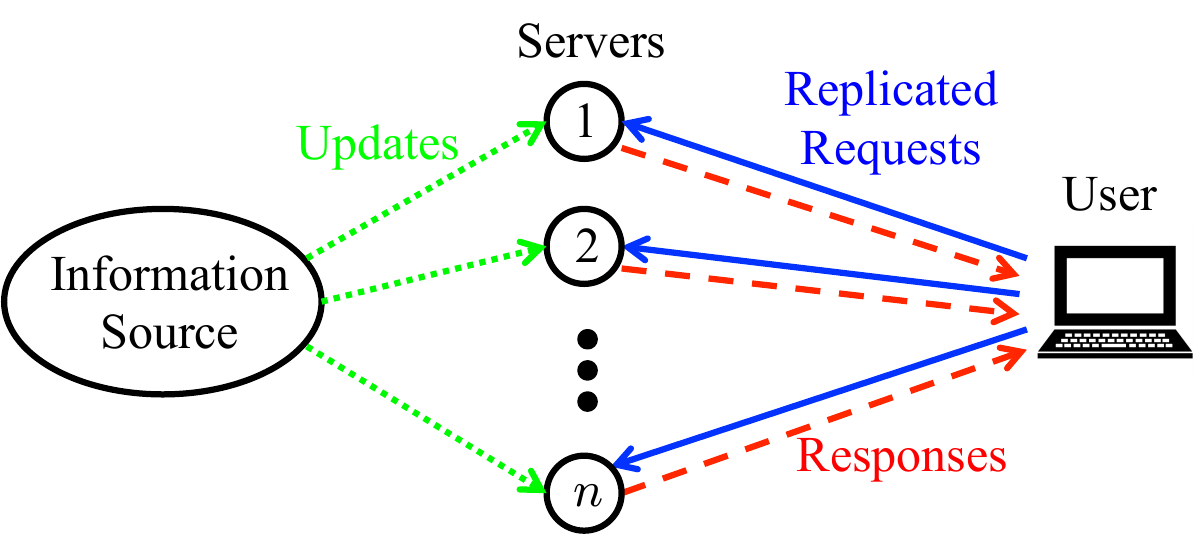}
\caption{The Pull model of information-update systems. 
\blue{Note that the arrows in the figure denote logical links rather than physical connections.
The updates, requests, and responses are all transmitted through (wired or wireless) networks.}}
\label{fig:model}
\end{figure}

\begin{figure}[!t]
\centering
\includegraphics[width=0.35\textwidth]{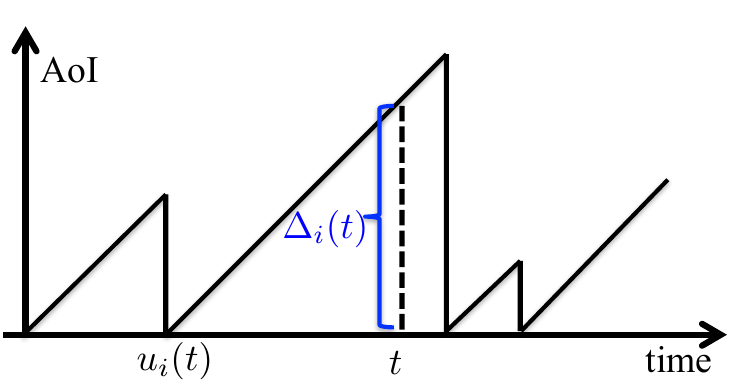}
\caption{An illustration of the AoI evolution at server $i$}
\label{fig:update}
\end{figure}

Under the $(n,k)$ replication scheme, when the user receives the first $k$ responses, it uses the freshest information 
among these $k$ responses to make certain decisions (e.g., stock trading decisions based on the received stock price 
information). Let $(j)$ denote the index of the server corresponding to the $j$-th response received by the user. 
Then, set $\mK \triangleq \{(1), (2), \dots, (k)\}$ contains the indices of the servers that return the first $k$ responses,
and the following is satisfied: $f_{(1)} \le f_{(2)} \le \dots \le f_{(k)}$ and $R_{(1)} \le R_{(2)} \le \dots \le R_{(k)}$.
Let server $i^*$ be the index of the server that provides the freshest information (i.e., that has the smallest AoI) among these $k$ responses
when downloading starts at time $s$, i.e., 
\begin{equation}
\label{eq:istar}
\Delta_{i^*}(s) = \min_{i \in \mK} \Delta_i(s).
\end{equation}
Here, we are interested in the AoI at the user's side when it receives the $k$-th response, denoted by $\Delta(k)$, 
which is the time difference between when the $k$-th response is received and when the information at 
server $i^*$ is updated, i.e., 
\begin{equation}
\label{eq:aoik}
\Delta(k) \triangleq f_{(k)} - u_{i^*}(s).
\end{equation}

Then, there are two natural questions of interest: 

\noindent \emph{(Q1): For a given $k$, can one obtain a closed-form formula 
for computing the expected AoI at the user's side, $\bE[\Delta(k)]$?} 

\noindent \emph{(Q2): How to determine the optimal number of responses 
to wait for, such that $\bE[\Delta(k)]$ is minimized?} 

The second question can be formulated as the following optimization problem:
\begin{equation}
\label{eq:opt}
\min_{k \in \mN} \bE\left[\Delta(k)\right].
\end{equation} 
We will answer these two questions in Section~\ref{sec:aoi}.

Furthermore, we will generalize the proposed framework and consider maximizing an AoI-based utility function at the user's side. The utility maximization problem will be studied in Section~\ref{sec:utility}, where we consider both cases of known and unknown system parameters (i.e., the updating rate and the mean response time).

\section{AoI Minimization}\label{sec:aoi}
In this section, we focus on the AoI minimization problem under the Pull model. We first derive a closed-form formula for computing the expected AoI at the user's side under the $(n,k)$ replication scheme (Section~\ref{sec:e_aoi}). Then, we find the optimal number of responses to wait for in order to minimize the expected AoI (Section~\ref{sec:aoi-opt}). Finally, we discuss some immediate extensions (Section~\ref{sec:extensions}).


\subsection{Expected AoI}\label{sec:e_aoi}

In this subsection, we focus on answering Question (Q1) and derive a closed-form formula for computing the expected AoI under the $(n,k)$ replication scheme.

To begin with, we provide a useful expression of the AoI at the user's side under the $(n,k)$ replication scheme (i.e., $\Delta(k)$, as defined in Eq.~\eqref{eq:aoik}) as follows: 
\begin{equation}
\begin{split}
\label{eq:aoik_l}
\Delta(k) &= f_{(k)} - u_{i^*}(s) \\
&= f_{(k)} - s + s -  u_{i^*}(s) \\
&= R_{(k)} + \Delta_{i^*}(s) \\
&= R_{(k)} + \min_{i \in \mK} \Delta_i(s),
\end{split}
\end{equation} 
where the second last equality is from the definition of $R_i$ and $\Delta_{i}(t)$ (i.e., Eq.~\eqref{eq:aoi}), and the last equality is from Eq.~\eqref{eq:istar}.
As can be seen from the above expression, under the $(n,k)$ replication scheme 
the AoI at the user's side consists of two terms:
(i) $R_{(k)}$, the total waiting time for receiving the first $k$ responses,  
and (ii) $\min_{i \in \mK} \Delta_i(s)$ (also denoted by $\Delta_{i^*}(s)$), the AoI of the freshest information 
among these $k$ responses when downloading starts at time $s$.
An illustration of these two terms and $\Delta(k)$ is shown in Fig.~\ref{fig:tradeoff}.

Taking the expectation of both sides of Eq.~\eqref{eq:aoik_l}, we have

\begin{equation}
\label{eq:e_aoi}
\bE[\Delta(k)] = \bE \left[R_{(k)} \right] + \bE \left [\min_{i \in \mK} \Delta_i(s) \right].
\end{equation}
In the above equation, the first term (i.e., the expected total waiting time) can be viewed as the cost of waiting, while the second term (i.e., the expected AoI of the freshest information among these $k$ responses) can be viewed as the benefit of waiting. Intuitively, as $k$ increases (i.e., waiting for more responses), the expected total waiting time (i.e., the first term) increases. On the other hand, upon receiving more responses, the expected AoI of the freshest information among these $k$ responses (i.e., the second term) decreases. 
Hence, there is a natural tradeoff between these two terms, which is a unique property of our 
newly introduced Pull model. 

Next, we formalize this tradeoff by deriving the closed-form expressions of the above two terms
as well as the expected AoI.
We state the main result of this subsection in Theorem~\ref{thm:e_aoi}.

\begin{figure}[!t]
	\centering
	\includegraphics[width=0.4\textwidth]{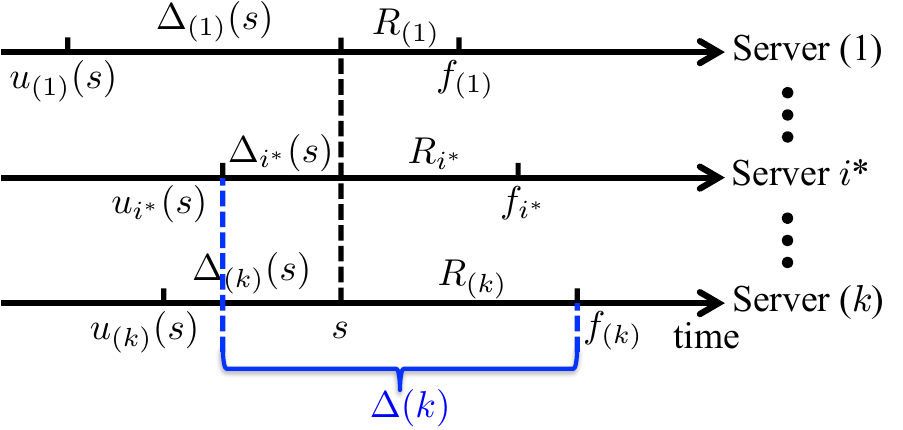}
	\caption{An illustration of the AoI at the user's side and its two terms under the ($n,k$) replication scheme}
	\label{fig:tradeoff}
\end{figure}

\begin{theorem}\label{thm:e_aoi}
	Under the $(n,k)$ replication scheme, the expected AoI at the user's side can be expressed as
	\begin{equation}
	\label{eq:e_aoi_formula}
	\bE[\Delta(k)] = \frac{1}{\nu}(\mathbf{H}(n) - \mathbf{H}(n-k)) + \frac{1}{k\lambda},
	\end{equation}
	where $\mathbf{H}(n) = \sum_{l=1}^n \frac{1}{l}$ is the $n$-th partial sum of the diverging harmonic series.
\end{theorem}

\begin{proof}
	\blue{We first analyze the first term of the right-hand side of Eq.~\eqref{eq:e_aoi} and want to show 
		$\bE[R_{(k)}] = \frac{1}{\nu}(\mathbf{H}(n) - \mathbf{H}(n-k))$.}
	Note that the response time is exponentially distributed with mean $1/\nu$ and is \emph{i.i.d.} across the servers.
	Hence, random variable $R_{(k)}$ is the $k$-th smallest value of $n$ \emph{i.i.d.} exponential random variables with mean $1/\nu$.
	The order statistics results of exponential random variables give that $R_{(j)} - R_{(j-1)}$ is an exponential random variable with mean $\frac{1}{(n+1-j)\nu}$ for any $j \in \mN$, where we set $R_{(0)}=0$ for ease of notation~\cite{firstcourse}.
	Hence, we have the following:
	\begin{equation}
	\begin{split}
	\bE \left[R_{(k)} \right] 
	&= \bE \left[ \sum_{j=1}^{k} (R_{(j)} - R_{(j-1)}) \right] \\
	&= \sum_{j=1}^{k} \bE \left[R_{(j)} - R_{(j-1)}\right] \\
	&= \sum_{j=1}^{k} \frac{1}{(n+1-j)\nu} \\
	&= \frac{1}{\nu}(\mathbf{H}(n) - \mathbf{H}(n-k)).\label{eq:e_R_k}
	\end{split}
	\end{equation}

	Next, we analyze the second term of the right-hand side of Eq.~\eqref{eq:e_aoi} and want to show the following:
	\begin{equation}
	\label{eq:e_Delta_i}
	\bE \left [\min_{i \in \mK} \Delta_i(s) \right] = \frac{1}{k\lambda}.
	\end{equation}
	Note that the updating process at the source for each server is a Poisson process with rate $\lambda$ and is \emph{i.i.d.} across the servers. We also assume that there is no transmission delay from the source to the servers.
	Hence, the inter-update time for each server is exponentially distributed with mean $1/\lambda$. 
	Due to the memoryless property of the exponential distribution, at any given request time $s$, the AoI at each server has the same distribution as the 
	inter-update time, i.e., random variable $\Delta_i(s)$ is also exponentially distributed with mean $1/\lambda$ 
	and is \emph{i.i.d.} across the servers \cite{nelson13}. Therefore, random variable $\min_{i \in \mK} \Delta_i(s)$
	is the minimum of $k$ \emph{i.i.d.} exponential random variables with mean $1/\lambda$, which is also exponentially
	distributed with mean $\frac{1}{k\lambda}$. This implies Eq.~\eqref{eq:e_Delta_i}.
	
	Combining Eqs.~\eqref{eq:e_R_k} and \eqref{eq:e_Delta_i}, we complete the proof.
\end{proof}

\emph{Remark.}
The above analysis indeed agrees with our intuition: while the expected total waiting time for receiving the first $k$ responses (i.e., Eq.~\eqref{eq:e_R_k})
is a monotonically increasing function of $k$, the expected AoI of the freshest information among these $k$ responses 
(i.e., Eq.~\eqref{eq:e_Delta_i}) is a monotonically decreasing function of $k$.

\subsection{Optimal Replication Scheme}\label{sec:aoi-opt}
In this subsection, we will exploit the aforementioned tradeoff and focus on answering Question (Q2) that we discussed at the end of Section~\ref{sec:model}.
Specifically, we aim to find the optimal number of responses to wait for in order to minimize the expected 
AoI at the user's side.

First, due to Eq.~\eqref{eq:e_aoi_formula}, we can rewrite the optimization problem in Eq.~\eqref{eq:opt} as
\begin{equation}
\label{eq:aoi_opt}
\min_{k \in \mN} \frac{1}{\nu}(\mathbf{H}(n) - \mathbf{H}(n-k)) + \frac{1}{k\lambda}.
\end{equation}
Let $k^*$ be an optimal solution to Eq.~\eqref{eq:aoi_opt}.
We state the main result of this subsection in Theorem~\ref{thm:aoi_opt}.

\begin{theorem}
	\label{thm:aoi_opt} 
	An optimal solution $k^*$ to Problem~\eqref{eq:aoi_opt} can be computed as
	\begin{equation}
	k^* = \min \left \{ \left \lceil \frac{2\nu n}{\sqrt{(\lambda + \nu)^2+4\lambda\nu n}+\lambda+\nu} \right \rceil, n  \right \}.
	\end{equation}
\end{theorem}

\begin{proof}
	We first define $D(k)$ as the difference of the expected AoI between the $(n,k+1)$ and $(n,k)$ 
	replication schemes, i.e., $D(k) \triangleq \bE[\Delta(k+1)] - \bE[\Delta(k)]$ for any $k \in \{1,2,\dots,n-1\}$.
	From Eq.~\eqref{eq:e_aoi_formula}, we have the following:
	\begin{equation}\label{eq:diff}
	D(k) = \frac{1}{(n-k)\nu} - \frac{1}{k(k+1)\lambda},
	\end{equation}
	for any $k \in \{1,2,\dots,n-1\}$.
	It is easy to see that $D(k)$ is a monotonically increasing function of $k$.
	
	We now extend the domain of $D(k)$ to the set of positive real numbers and want to find $k^{\prime}$
	such that $D(k^{\prime}) = 0$. With some standard calculations and dropping the negative solution, 
	we derive the following:
	\begin{equation}
	k^{\prime} = \frac{2\nu n}{\sqrt{(\lambda + \nu)^2+4\lambda\nu n}+\lambda+\nu}.
	\end{equation}
	Next, we discuss two cases: (i) $k^{\prime}>n-1$ and (ii) $0 < k^{\prime} \le n-1$.
	
	In Case (i), we have $k^{\prime}>n-1$. This implies that $D(k)=\bE[\Delta_i(k+1)]-\bE[\Delta_i(k)]<0$ for all $k \in \{1,2,\dots,n-1\}$, as $D(k)$ is monotonically increasing. Hence, the expected AoI, $\bE[\Delta(k)]$, is a monotonically decreasing function for $k \in \mN$.
	Therefore, $k^*=n$ must be the optimal solution.
	
	In Case (ii), we have $0< k^{\prime} \le n-1$. We consider two subcases: $k^{\prime}$ is an integer in $\{1,2,\dots,n-1\}$ and $k^{\prime}$ is not an integer.
	
	If $k^{\prime}$ is an integer in $\{1,2,\dots,n-1\}$, we have $D(k)=\bE[\Delta(k+1)]-\bE[\Delta(k)] \leq 0$ for $k \in \{1,2,\dots,k^{\prime}\}$ and $D(k)=\bE[\Delta(k+1)]-\bE[\Delta(k)]>0$ 
		for $k\in\{k^{\prime}+1,\dots,n-1\}$, as $D(k)$ is monotonically increasing.
		Hence, the expected AoI, $\bE[\Delta(k)]$, is first decreasing (for $k \in \{1,2,\dots,k^{\prime}\}$) and then increasing (for $k\in\{k^{\prime}+1,\dots,n\}$).
		Therefore, there are two optimal solutions: $k^*=k^{\prime}$ and $k^*=k^{\prime}+1$ since $\bE[\Delta(k^{\prime}+1)] = \bE[\Delta(k^{\prime})]$ (due to $D(k^{\prime})=0$).
	
	If $k^{\prime}$ is not an integer, we have $D(k)=\bE[\Delta(k+1)]-\bE[\Delta(k)]<0$ for $k \in \{1,2,\dots,\lfloor k^{\prime} \rfloor\}$ and $D(k)=\bE[\Delta(k+1)]-\bE[\Delta(k)]>0$ 
	for $k\in\{\lceil k^{\prime} \rceil,\dots,n-1\}$, as $D(k)$ is monotonically increasing. 
	Hence, the expected AoI, $\bE[\Delta(k)]$, is first decreasing (for $k \in \{1,2,\dots,\lfloor k^{\prime} \rfloor, \lceil k^{\prime} \rceil \}$) and then increasing (for $k\in\{\lceil k^{\prime} \rceil,\dots,n-1\}$).
	Therefore, $k^* = \lceil k^{\prime} \rceil$ must be the optimal solution.
	
	Combining two subcases, we have $k^* = \lceil k^{\prime} \rceil$ in Case (ii). Then, combining Cases (i) and (ii), we have 
	$k^* = \min \{\lceil k^{\prime} \rceil, n\} =  \min \{ \lceil \frac{2\nu n}{\sqrt{(\lambda + \nu)^2+4\lambda\nu n}+\lambda+\nu} \rceil, n \}$.
\end{proof}

\emph{Remark.} There are two special cases that are of particular interest: 
(i) waiting for the first response only (i.e., $k^*=1$) and (ii) waiting for all the responses (i.e., $k^*=n$).
In Corollary~\ref{cor:aoi_opt_special}, we provide a sufficient and necessary condition for each of these two special cases.

\begin{corollary}
	\label{cor:aoi_opt_special}
	(i) $k^*=1$ is an optimal solution to Problem~\eqref{eq:aoi_opt} if and only if $\lambda \geq \frac{\nu(n-1)}{2}$;
	(ii) $k^*=n$ is an optimal solution to Problem~\eqref{eq:aoi_opt} if and only if $\lambda \leq \frac{\nu}{n(n-1)}$.
\end{corollary}

\begin{proof}
	The proof follows straightforwardly from Theorem~\ref{thm:aoi_opt}. 
	A little thought gives the following: $k^*=1$ is an optimal solution if and only if $D(1) \ge 0$. 
	Solving $D(1) = \frac{1}{(n-1)\nu} - \frac{1}{2\lambda} \ge 0$ gives $\lambda \geq \frac{\nu(n-1)}{2}$.
	Similarly, $k^*=n$ is an optimal solution if and only if $D(n-1) \le 0$. 
	Solving $D(n-1) = \frac{1}{\nu} - \frac{1}{n(n-1)\lambda} \le 0$ gives $\lambda \leq \frac{\nu}{n(n-1)}$.
\end{proof}

\emph{Remark.} The above results agree well with the intuition.
For a given number of servers, if the inter-update time is much smaller than the response time (i.e., $1/\lambda \ll 1/\nu$), then the difference of the freshness levels among the servers is relatively small. In this case, it is not beneficial to wait for more responses.
On the other hand, if the inter-update time is much larger than the response time (i.e., $1/\lambda \gg 1/\nu$), 
then one server may possess much fresher information than another server. In this case, it is worth waiting for more responses, which leads to a significant gain in the AoI reduction.

Note that Theorem~\ref{thm:aoi_opt} also implies how the optimal solution (i.e., $k^*$) scales as the number of servers (i.e., $n$) increases: when $n$ becomes large, we have $k^* = \lceil \frac{2\nu n}{\sqrt{(\lambda + \nu)^2+4\lambda\nu n}+\lambda+\nu} \rceil = O(\sqrt{n})$.

\subsection{Extensions}\label{sec:extensions}
In this subsection, we discuss some immediate extensions of the considered model, including more general replication schemes and different types of response time distributions.

\subsubsection{Replication schemes}
So far, we have only considered the $(n,k)$ replication scheme. 
\blue{One limitation of this scheme is that it requires the user to send a replicated request to every server, which may incur a large overhead when there are a large number of servers (i.e., when $n$ is large). }
Instead, a more practical scheme would be to send the replicated requests to a subset of servers.
Hence, we consider the $(n,m,k)$ replication schemes, under which the user sends a replicated 
request to each of the $m$ servers that are chosen from the $n$ servers uniformly at random 
and waits for the first $k$ responses, where $m \in \mN$ and $k \in \{1,2,\dots,m\}$. 
Making the same assumptions as in Section~\ref{sec:model}, we can derive the expected AoI at 
the user's side in a similar manner. Specifically, reusing the proof of Theorem~\ref{thm:e_aoi} and
replacing $n$ with $m$ in the proof, we can show the following:
\begin{equation}
\bE[\Delta(k)] = \frac{1}{\nu}(\mathbf{H}(m) - \mathbf{H}(m-k)) + \frac{1}{k\lambda}.
\end{equation}

\subsubsection{Uniformly distributed response time}
Note that our current analysis requires the memoryless property of the Poisson updating process. 
However, the analysis can be extended to the uniformly distributed response time. We make the 
same assumptions as in Section~\ref{sec:model}, except that the response time is now uniformly 
distributed on interval $[b,b+h]$ with $b \ge 0$ and $h \ge 0$.
In this case, it is easy to derive $\bE[R_{(k)}] = \frac{kh}{n+1} + b$ (see, e.g., \cite{firstcourse}).
Since Eq.~\eqref{eq:e_Delta_i} still holds, from Eq.~\eqref{eq:e_aoi} we have
\begin{equation}
\label{eq:e_aoi_uniform}
\bE[\Delta(k)] = \frac{kh}{n+1} + b + \frac{1}{k\lambda}.
\end{equation}

Following a similar line of analysis to that in the proof of Theorem~\ref{thm:aoi_opt},
we can show that an optimal solution $k^*$ can be computed as
\begin{equation}
k^* =\min \left\{ \left\lceil \frac{2(n+1)}{\sqrt{h^2\lambda^2+4h\lambda(n+1)} + h\lambda} \right \rceil, n \right\}.
\end{equation}

\subsubsection{Heterogeneous servers} In order to obtain the theoretical results and corresponding insights, we have assumed homogeneous servers in our model. However, it is important to consider realistic settings with heterogeneous servers. That is, the servers have different mean inter-update times and different mean response times. In the following, we share our thoughts about the extension of our analysis to the settings with heterogeneous servers and discuss the challenges.
Recall from Eq.~\eqref{eq:e_aoi} that the expected AoI at the user's side consists of two terms: (i) $\bE[R_{(k)}]$, the expected total waiting time for receiving the first $k$ responses, and (ii) $\bE[\min_{i \in \mK} \Delta_i(s)]$, the expected AoI of the freshest information among these $k$ responses at request time $s$.
Following a similar line of analysis to that in the proof of Theorem~\ref{thm:e_aoi} and applying the order statistics results for independent and non-identically distributed exponential random variables, it is not difficult to derive the expression for $\mathbb{E}[R_{(k)}]$, which is more involved though. On the other hand, it becomes much harder to derive the closed-form expression for  $\mathbb{E}[\min_{i\in\mK} \Delta_i(s)]$ as the analysis involves $\binom{n}{k}$ possible combinations for the realization of the first $k$ responses and the probability of each realization depends on the mean response times. Therefore, it becomes more challenging to derive the closed-form expression for the expected AoI and thus the optimal solution $k^*$ in such heterogeneous settings.

\vskip 0.5cm

\section{AoI-based Utility Maximization}\label{sec:utility}
In Section~\ref{sec:aoi}, our study has been focused on minimizing the expected AoI at the user's side. For certain practical applications, however, the user might be more interested in maximizing the utility that is dependent on the AoI than minimizing the AoI itself. Such an AoI-based utility function can serve as a \emph{Quality of Experience (QoE)} metric, which measures the user's satisfaction level with respect to freshness of the received information. To that end, in this section we will investigate the problem of AoI-based utility maximization. Specifically, we will consider both cases of known and unknown system parameters (i.e., the updating rate and the mean response time) in Sections~\ref{sec:utility_known} and \ref{sec:utility_unknown}, respectively.

\subsection{AoI-based Utility Function}\label{sec:utility_def}
Consider a function $U: [0, \infty) \rightarrow [0, \infty)$, which maps the AoI at the user's side under the $(n, k)$ replication scheme (i.e., $\Delta(k)$) to a utility obtained by the user. Such a function $U(\cdot)$ is called a utility function. Similar to \cite{sun17update}, we assume that the utility function $U(\cdot)$ is measurable, non-negative, and non-increasing. The specific choice of the utility function depends on applications under consideration in practice.


We consider the same model as that in Section~\ref{sec:model}.
From the analysis in the proof of Theorem~\ref{thm:e_aoi}, it is easy to see that the AoI, $\Delta(k)$, is the sum of $k+1$ independent exponential random variables, which are $R_{(j)} - R_{(j-1)}$ for $j \in \{1, 2, \dots, k\}$, and $\min_{i \in \mK} \Delta_i(s)$. Therefore, the AoI, $\Delta(k)$, is a hyperexponential random variable (or a generalized Erlang random variable). The probability density function of a hyperexponential random variable with rate parameters $\alpha_1, \alpha_2, \dots, \alpha_r$ can be expressed as $f(x) = \sum_{i=1}^r  w_i \alpha_i e^{-\alpha_i x}$, where $\alpha_i$ is the rate of the $i$-th exponential distribution and $w_i=\prod_{j=1,j\neq i}^r \frac{\alpha_j}{\alpha_j-\alpha_i}$. For the AoI, $\Delta(k)$, we have $r=k+1$, and the rate parameters $\alpha_i$'s are: $\alpha_i = (n+1-i) \nu$ for $i=1,2,\dots,k$ and $\alpha_{k+1} = k \lambda$.
Then, the expected utility can be calculated as
\begin{equation}
\label{eq:e_utility}
\bE[U(\Delta(k))] = \int_0^{\infty} U(x) \sum_{i=1}^{k+1} w_i \alpha_i e^{-\alpha_i x} dx.
\end{equation}
Now, the problem is to find the optimal value $k^*$ that achieves the maximum expected utility:
\begin{equation}
\label{eq:utility_opt}
k^* \in \argmax_{k \in \mN} \bE[U(\Delta(k))].
\end{equation}

In the following subsections, we will consider a specific AoI-based utility function, which is an exponential function of the negative AoI, and aim to maximize the expected utility. We will consider both cases of known and unknown system parameters (i.e., the updating rate and the mean response time).

\subsection{Case with Known System Parameters}\label{sec:utility_known}
In this subsection, we consider a specific AoI-based utility function in the following exponential form:
\begin{equation}\label{eq:utility_exp}
U(\Delta(k)) = e^{-a\Delta(k)},
\end{equation}
where $a$ is a positive constant.
The above exponential utility function implies that the user receives the full utility when the AoI is zero (which is an ideal case) and the utility decreases exponentially as the AoI increases. Such a utility function decreases very quickly with respect to the AoI and is desirable for real-time applications that require extremely fresh information to provide satisfactory service to the users (e.g., stock quote service).

Assuming that the updating rate and the mean response time are known, we first derive a closed-form formula for computing the expected utility $\bE[U(\Delta(k))]$. Then, we find an optimal $k^*$ that yields the maximum expected utility. The main results of this subsection are stated in Theorems~\ref{thm:e_utility} and \ref{thm:utility_opt}.
The proofs of Theorems~\ref{thm:e_utility} and \ref{thm:utility_opt} follow a similar line of analysis to that for Theorems~\ref{thm:e_aoi} and \ref{thm:aoi_opt}, respectively. 
The detailed proofs are provided in Appendices~\ref{app:e_utility} and \ref{app:utility_opt}, respectively.



\begin{theorem}\label{thm:e_utility}
Under the $(n,k)$ replication scheme, the expected utility can be expressed as
\begin{equation}
\label{eq:e_utility_formula}
\bE[U(\Delta(k))] =  \frac{k\lambda}{k\lambda+a}\prod_{j=1}^k \frac{(n+1-j)\nu}{(n+1-j)\nu+a}.
\end{equation}
\end{theorem}

\begin{theorem}
\label{thm:utility_opt} 
An optimal solution $k^*$ to Problem~\eqref{eq:utility_opt} (i.e., achieving the maximum expected utility) can be computed as
\begin{equation}
k^* = \min \left \{ \left \lceil \frac{2\nu n}{\sqrt{(\lambda + \nu+a)^2+4\lambda\nu n}+(\lambda+\nu+a)} \right \rceil, n  \right \}.
\end{equation}
\end{theorem}


\emph{Remark.} Similar to the AoI minimization problem studied in Section~\ref{sec:aoi-opt}, there are also two interesting special cases: 
(i) waiting for the first response only (i.e., $k^*=1$) and (ii) waiting for all the responses (i.e., $k^*=n$).
In Corollary~\ref{cor:utility_opt_special}, we provide a sufficient and necessary condition for each case. 

\begin{corollary}
\label{cor:utility_opt_special}
(i) $k^*=1$ is an optimal solution to Problem~\eqref{eq:utility_opt} if and only if $\lambda \geq \frac{\nu(n-1)}{2}-\frac{a}{2}$;
(ii) $k^*=n$ is an optimal solution to Problem~\eqref{eq:utility_opt} if and only if $\lambda \leq \frac{\nu}{n(n-1)} - \frac{a}{n}$.
\end{corollary}

The proof of Corollary~\ref{cor:utility_opt_special} is provided in Appendix~\ref{app:utility_opt_special}.

\subsection{Case with Unknown System Parameters}\label{sec:utility_unknown}
In Section~\ref{sec:utility_known}, we have addressed the utility maximization problem in Eq. \eqref{eq:utility_opt}, assuming the knowledge of the updating rate (i.e., $\lambda$) and the mean response time (i.e., $1/\nu$). Similar assumptions are also made for obtaining a good understanding of the studied theoretical problems (see, e.g., \cite{aoisurvey,sun17update} and references therein). However, such information is typically unavailable to the user in practice. For example, the user generally has no prior knowledge of the updating processes between the information source and the servers. Moreover, it is difficult, if not impossible, for the user to estimate the updating rate as the user has no direct observation about the updating processes. Therefore, an interesting and important question naturally arises: \emph{How to maximize the expected utility in the presence of unknown system parameters?} 

To that end, in this subsection we aim to address the above question through the design of learning-based algorithms. Specifically, in the presence of unknown system parameters we will reformulate the utility maximization problem as a stochastic \emph{Multi-Armed Bandit (MAB)} problem. 
\emph{To the best of our knowledge, this is one of the first studies that leverage an MAB formulation to study the AoI problem.}


In the following, we will first briefly introduce the basic setup of the stochastic MAB model (Section~\ref{sec:mab_intro}). Then, we formulate the utility maximization problem with unknown system parameters as an MAB problem and explain the special linear feedback graph of our problem, which can be exploited to achieve improved performance guarantees (Section~\ref{sec:mab_formulation}). Finally, we introduce several MAB algorithms that can be applied to address our problem (Section~\ref{sec:mab_alg}).


\subsubsection{The MAB Model}
\label{sec:mab_intro}

The MAB model has been widely employed for studying many sequential decision-making problems of practical importance (clinical trials, network resource allocation, online ad placement, crowdsourcing, etc.) with unknown parameters (see, e.g., \cite{lai1985asymptotically,gittins2011multi,bubeck2012regret,li2019infocom}).

%
In the classic MAB model, a player (i.e., a decision maker) is faced with $n$ options, which are often called arms in the MAB literature. 
In each round, the player can choose to play one arm and receives the reward generated by the played arm. The reward of playing arm $k$ in round $t$, denoted by $X_{k,t}$, is a random variable distributed on interval $[0,1]$, i.e., $X_{k,t} \in [0,1]$. 
The reward $X_{k,t}$ of each arm $k$ is assumed to be \emph{i.i.d.} over time.
Let $\mu_k$ be the mean reward of arm $k$; let $\mu^*$ be the highest mean reward among all the arms, i.e., $\mu^* \triangleq \max_k \mu_k$. The specific distributions of $X_{k,t}$'s and the values of $\mu_k$'s are unknown to the player.

An algorithm $\pi$ chooses an arm $I_t$ to play in each round $t \in \{1,2,\dots,T\}$, where $T$ is the length of the time horizon.
The objective here is to design an algorithm that maximizes the expected cumulative reward during this time horizon, i.e., $\sum_{t=1}^{T} \mu_{I_t}$. This is equivalent to minimizing the \emph{regret}, which is the difference between the expected cumulative reward obtained by an optimal algorithm that always plays the best arm and that of the considered algorithm. We use $R(T)$ to denote the regret, which is formally defined as follows:
\begin{equation}\label{eq:regret}
R(T) \triangleq \mu^*T - \sum_{t=1}^{T} \mu_{I_t}.
\end{equation}

In order to maximize the reward or minimize the regret, the player is faced with a key tradeoff:
how to balance \emph{exploitation} (i.e., playing the arm with the highest empirical mean reward) and \emph{exploration} (i.e., trying other arms, which could potentially be better)? There exist several well-known algorithms that can address this challenge. We will discuss them in Section~\ref{sec:mab_alg}.

\subsubsection{The MAB Formulation of the Utility Maximization Problem}\label{sec:mab_formulation}

We now want to formulate the utility maximization problem with unknown system parameters as an MAB problem.
Note that when the updating rate and the mean response time are unknown, one cannot easily derive a closed-form formula for the expected utility and find the optimal solution as in Section~\ref{sec:utility_known}. Therefore, for each sent request the user needs to decide how many responses to wait for in a dynamic manner. In this case, one can naturally reformulate the utility maximization problem using the MAB model: making a decision for each sent request corresponds to a round; waiting for $k$ responses corresponds to playing arm $k$.  
Let $\Delta(k,t)$ be the AoI at the user's side when the user sends the $t$-th request and waits for the first $k$ responses. Then, the utility $U(\Delta(k,t))$, normalized to interval $[0,1]$, corresponds to the obtained reward $X_{k,t}$ of playing arm $k$ in round $t$. The mean reward of arm $k$ is $\mu_k = \bE[X_{k,t}] = \bE[U(\Delta(k,t)]$. In this MAB formulation, the utility function $U(\cdot)$ is not limited to the exponential function in the form of Eq.~\eqref{eq:utility_exp}; instead, $U(\cdot)$ can be very general, as long as it is a measurable, non-negative, and non-increasing of the AoI. 
Note that for each arm $k$, the reward $X_{k,t}=U(\Delta(k,t))$ is \emph{i.i.d.} over rounds since the AoI $\Delta(k,t)$ is \emph{i.i.d.} over rounds due to the memoryless property of the exponential distribution.
We provide a detailed explanation for this in Appendix~\ref{app:iid_explanation}.


 

Recently, MAB models with side observations have been studied (see, e.g., \cite{mannor2011bandits,caron2012leveraging,buccapatnam2017reward}). In these models, playing an arm not only reveals the reward of the played arm but also that of some other arm(s). Such side observations are typically encoded in a feedback graph, where each node corresponds to an arm and each directed edge $(k,k^{\prime})$ means that playing arm $k$ also reveals the reward of arm $k^{\prime}$.

We would like to point out that the utility maximization problem with unknown system parameters can be formulated as an MAB problem with side observations. 
Note that although the rewards of different arms are dependent in our problem, as pointed out in \cite{ucb1}, the MAB formulation and various learning algorithms are still applicable in such settings.
Moreover, as illustrated in Fig.~\ref{fig:directedgraph}, such dependence leads to a special linear structure in the associated feedback graph of our problem.
Specifically, note that upon receiving the $k$-th response, the user has the information about the first $k-1$ responses. Thus, the user can know the utility she would have obtained if she had waited for only $k^{\prime}$ responses for all $k^{\prime} < k$. Mapping this property to the MAB model, it means that playing arm $k$ reveals not only the reward of arm $k$ but also that of arm $k^{\prime}$ for all $k^{\prime} < k$.
Such special properties can be leveraged to design learning algorithms that perform exploration more efficiently and thus lead to an improved regret performance.

\begin{figure}[!t]
\centering
\includegraphics[width=0.35\textwidth]{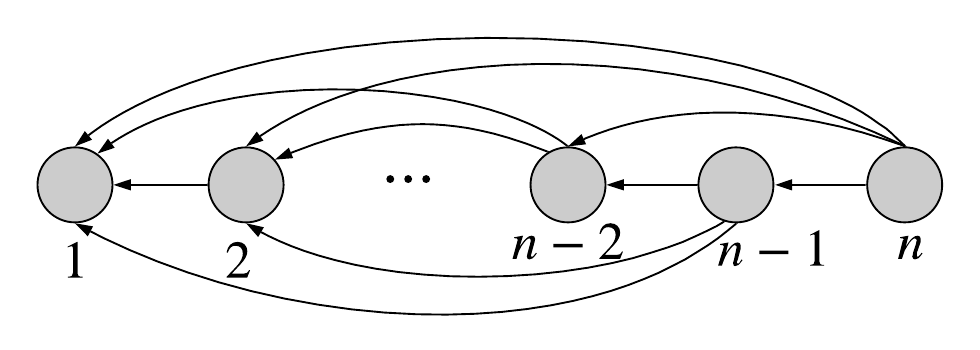}
\caption{Linear feedback graph where each node $k \in \{2,3,\dots,n\}$ has a directed edge to every node in $\{1,2,\dots,k-1\}$ and node 1 does not have any outgoing edge}
\label{fig:directedgraph}
\end{figure}

\begin{algorithm}[!t] 
	\caption{$\epsilon_t$-Greedy \cite{ucb1}}
	\begin{algorithmic}[1]
	    \label{alg:greedyN}
		\STATE Input: $c>0$ and $0 < d < 1$.
		\STATE Let $\bar{x}_k(t)$ be the empirical mean reward of arm $k$ at the beginning of round $t$. 
		\FOR{$t = 1, 2, \dots, T$}
		\STATE Let $j^* \in \argmax_k \bar{x}_k(t)$ and $\epsilon_t \triangleq \min \{1,\frac{cn}{d^2t} \}$.
		\STATE Play arm $j^*$ with probability $1 - \epsilon_t$; play a random arm (uniformly) with probability $\epsilon_t$.
		\STATE Update $\bar{x}_{I_t}(t+1)$ for the played arm $I_t$. \label{line:greedy-update}
		\ENDFOR
	\end{algorithmic}
\end{algorithm}

\begin{algorithm}[!t]
    \caption{$\epsilon_t$-Greedy-LP \cite{buccapatnam2017reward}}
	\begin{algorithmic}[1]\label{alg:greedylp}
		\STATE Input: $c>0$ and $0 < d < 1$.
		\STATE Let $\bar{x}_k(t)$ be the empirical mean reward of arm $k$ at the beginning of round $t$. 
		\FOR{$t = 1, 2, \dots, T$}
		\STATE Let $j^* \in \argmax_k \bar{x}_k(t)$ and $\epsilon_t \triangleq \min \{1,\frac{cn}{d^2t} \}$.
		\STATE Play arm $j^*$ with probability $1 - \epsilon_t$; play arm $n$ with probability $\epsilon_t$.
		\STATE Update $\bar{x}_k(t+1)$ for all $k \in \{1,\cdots, n\}$, accounting for all the observations, including side observations.
		\ENDFOR
	\end{algorithmic}
\end{algorithm}

\subsubsection{Learning Algorithms}\label{sec:mab_alg}

There exist several well-known learning algorithms that can address the classic MAB problem, including \emph{$\epsilon_t$-Greedy} \cite{ucb1} and \emph{Upper Confidence Bound (UCB)} \cite{lai1985asymptotically,ucb1,iucb}. 
In the sequel, we will introduce these algorithms and explain how to leverage the side observations and the special linear structure of the graphical feedback to design algorithms with improved regret upper bounds.

We begin with $\epsilon_t$-Greedy, a very simple algorithm that performs exploration explicitly. Specifically, it plays the arm with the highest empirical mean reward with probability $1-\epsilon_t$ (i.e., exploitation) and plays a random arm (uniformly) with probability $\epsilon_t$ (i.e., exploration), where $\epsilon_t$ decreases as $O(1/t)$. 
We summarize the operations of $\epsilon_t$-Greedy in Algorithm~\ref{alg:greedyN}.
When side observations are available, one can incorporate additional samples from side observations into the update of the empirical mean reward of non-played arms. We call $\epsilon_t$-Greedy that exploits the side observations as \emph{$\epsilon_t$-Greedy-N}. 
Note that $\epsilon_t$-Greedy-N is almost the same as $\epsilon_t$-Greedy except that in Line~\ref{line:greedy-update} of Algorithm~\ref{alg:greedyN}, one needs to update the empirical mean reward $\bar{x}_k(t+1)$ for all arms $k \in \{1,\cdots, n\}$, including the non-played arms, accounting for side observations. 
Apparently, $\epsilon_t$-Greedy-N accelerates the exploration process by taking advantage of additional samples from side observations and is expected to outperform $\epsilon_t$-Greedy.

Although $\epsilon_t$-Greedy-N leverages side observations and can speed up the exploration process compared to $\epsilon_t$-Greedy, it still randomly chooses an arm during the exploration process, being agnostic about the structure of the feedback graph. Therefore, all the arms have to be played in the exploration phase (i.e., with probability $\epsilon_t$). The analysis in \cite{ucb1} suggests that $O(\log{T})$ samples are sufficient for accurately estimating the mean reward of an arm. This implies that both $\epsilon_t$-Greedy and $\epsilon_t$-Greedy-N have a regret upper bounded by $O(n \log{T})$. However, many of such explorations appear unnecessary in our studied problem. This is because in our utility maximization problem, playing arm $n$ reveals a sample for every arm, due to the special linear structure of the feedback graph in Fig.~\ref{fig:directedgraph}. This suggests that one should always choose arm $n$ for exploration, which leads to a graph-aware algorithm summarized in Algorithm~\ref{alg:greedylp}. We call this algorithm \emph{$\epsilon_t$-Greedy-LP} as it turns out to be a special case of the $\epsilon_t$-Greedy-LP algorithm proposed in \cite{buccapatnam2017reward}.

One can show that the regret of $\epsilon_t$-Greedy-LP is upper bounded by $O(\log{T})$, which improves upon $O(n \log{T})$ of $\epsilon_t$-Greedy and $\epsilon_t$-Greedy-N. This result follows immediately from Corollary~8 of \cite{buccapatnam2017reward} as the linear feedback graph in Fig.~\ref{fig:directedgraph} is a special case of the graphs considered in \cite{buccapatnam2017reward}.
Note that the improved regret upper bound relies on the assumption that $\epsilon_t$-Greedy-LP has the knowledge of the difference between the reward of the optimal arm and that of the best suboptimal arm for choosing parameters $d$ and $c$ (see Corollary~8 of \cite{buccapatnam2017reward} for the specific form).

\begin{algorithm}[!t]
	\caption{UCB1 \cite{ucb1}}
	\begin{algorithmic}[1]\label{alg:ucb1}
		\STATE Let $\bar{x}_k(t)$ and $T_k(t)$ be the empirical mean reward and the total number of samples of arm $k$ at the beginning of round $t$, respectively. 
		\FOR{$t = 1, 2, \dots, T$}
		\STATE Play arm $I_t$ such that 
		\begin{equation}\label{eq:ucb1}
		I_t \in  \argmax_{k} \left \{ \bar{x}_k(t) + \sqrt{\frac{2\log{t}}{T_k(t)}} \right \}.
		\end{equation}
		\STATE Update $\bar{x}_{I_t}(t+1)$ and $T_{I_t}(t+1)$. \label{line:ucb-update}
		\ENDFOR
	\end{algorithmic}
\end{algorithm}

\begin{algorithm}[!t]
	\caption{UCB-LP \cite{buccapatnam2017reward}}
	\begin{algorithmic}[1]\label{alg:ucblp}
		\STATE \textbf{Initialization:} Set $B_0 = \mN$ and $\tilde{\delta}_0 = 1$.
		\STATE Let $\bar{x}_k(m)$ and $T_k(m)$ be the empirical mean reward and the total number of samples of arm $k$ up to and including stage $m$.
		\FOR{$m = 0, 1, 2, \dots, \left \lfloor \frac{1}{2} \log_2{\frac{T}{e}}\right \rfloor$}
		\STATE \textbf{Arm selection:} 
		\STATE Let $t_m \triangleq \left \lceil \frac{2\log{(T\tilde{\delta}_m^2)}}{\tilde{\delta}_m^2} \right \rceil $. 
		\IF{$|B_m| = 1$} 
		\STATE Play the single arm in $B_m$ until time $T$.
		\ELSIF{$2|B_m|\tilde{\delta}_m \geq 1$} \label{line:ucblp-explore-start}
		\STATE Play arm $n$ for $(t_m - t_{m-1})$ times.
		\ELSE
		\STATE Play each arm $k \in B_m$ for $(t_m - t_{m-1})$ times.
		\ENDIF \label{line:ucblp-explore-end}
		\STATE Update $\bar{x}_k(m)$ and $T_k(m)$ for all $k \in \{1,\cdots,n\}$.
		\STATE \textbf{Arm elimination:}
		\STATE Let $D_m$ be the set of all arms $j$ in $B_m$ for which
		\begin{equation}\label{eq:elimination}
		\bar{x}_j(m) + \sqrt{\frac{\log{(T\tilde{\delta}_m^2)}}{2T_j(m)}} < \max_{k\in B_m} \left \{ \bar{x}_k(m) - \sqrt{\frac{\log{(T\tilde{\delta}_m^2)}}{2T_k(m)}} \right \}.
		\end{equation}
		\STATE Set $B_{m+1} = B_m \setminus D_m$ and $\tilde{\delta}_{m+1}  =  \tilde{\delta}_m/2$.
		\ENDFOR
	\end{algorithmic}
\end{algorithm}

\begin{algorithm}[!t]
	\caption{UCB-LFG}
	\begin{algorithmic}[1]\label{alg:ucblfg}
		\STATE \textbf{Initialization:} Set $B_0 = \mN$ and $\tilde{\delta}_0 = 1$.
		\STATE Let $\bar{x}_k(m)$ and $T_k(m)$ be the empirical mean reward and the total number of samples of arm $k$ up to and including stage $m$.
		\FOR{$m = 0, 1, 2, \dots, \left \lfloor \frac{1}{2} \log_2{\frac{T}{e}}\right \rfloor$}
		\STATE \textbf{Arm selection:} 
		\STATE Let $t_m \triangleq \left \lceil \frac{2\log{(T\tilde{\delta}_m^2)}}{\tilde{\delta}_m^2} \right \rceil $. 
		\IF{$|B_m| = 1$} 
		\STATE Play the single arm in $B_m$ until time $T$.
		\ELSE
		\STATE Play arm $j^*(B_m)$ for $(t_m - t_{m-1})$ times, where $j^*(B_m)$ is the largest index of arms in set $B_m$. \label{line:ucblfg-explore}
		\ENDIF
		\STATE Update $\bar{x}_k(m)$ and $T_k(m)$ for all $k \in \{1,\cdots, n\}$.
		\STATE \textbf{Arm elimination:}
		\STATE Let $D_m$ be the set of all arms $j$ in $B_m$ for which
		\begin{equation}\label{eq:ucblfg-elimination}
		\bar{x}_j(m) + \sqrt{\frac{\log{(T\tilde{\delta}_m^2)}}{2T_j(m)}} < \max_{k\in B_m} \left \{ \bar{x}_k(m) - \sqrt{\frac{\log{(T\tilde{\delta}_m^2)}}{2T_k(m)}} \right \}.
		\end{equation}
		\STATE Set $B_{m+1} = B_m \setminus D_m$ and $\tilde{\delta}_{m+1}  =  \tilde{\delta}_m/2$.
		\ENDFOR
	\end{algorithmic}
\end{algorithm}

Next, we consider another simple algorithm called Upper Confidence Bound (UCB). As the name suggests, UCB considers the upper bound of a suitable confidence interval for the mean reward of each arm and chooses the arm with the highest such upper confidence bound (see, e.g., Eq.~\eqref{eq:ucb1}). There are several variants of the UCB algorithm \cite{lai1985asymptotically,ucb1,iucb}. We present a popular variant, called \emph{UCB1}, in Algorithm~\ref{alg:ucb1}. When side observations are available, similar to $\epsilon_t$-Greedy-N, there is a slightly modified UCB algorithm, called \emph{UCB-N} \cite{caron2012leveraging}, which incorporates additional samples from side observations into the update of the empirical mean reward and the total number of samples of non-played arms (i.e., Line~\ref{line:ucb-update} in Algorithm~\ref{alg:ucb1}).
Like $\epsilon_t$-Greedy-N, UCB-N is also agnostic about the structure of the feedback graph. In order to take the graph structure into consideration, we introduce \emph{UCB-LP}, which is based on another UCB variant, called UCB-Improved~\cite{iucb}. UCB-LP is a special case of the one proposed in \cite{buccapatnam2017reward}. We summarize the operations of UCB-LP in Algorithm~\ref{alg:ucblp}.

The key idea of UCB-LP is the following: we divide $T$ into multiple stages. For each stage $m$, we use $B_m$ to denote the set of arms not eliminated yet and use $\tilde{\delta}_m$ to estimate $\delta_k$.
At the beginning, set $B_0$ is initialized to the set of all arms; the value of $\tilde{\delta}_0$ is initialized to 1. We ensure that by the end of stage $m$, there are at least $t_m$ samples available for each arm in set $B_m$, from  playing either arm $n$ or the arm itself, where $t_m$ is determined by $\tilde{\delta}_m$ (Lines~\ref{line:ucblp-explore-start}-\ref{line:ucblp-explore-end}). Then, at the end of stage $m$ we obtain set $B_{m+1}$ by eliminating those arms estimated to be suboptimal according to Eq.~\eqref{eq:elimination} and obtain $\tilde{\delta}_{m+1}$ by halving the value of $\tilde{\delta}_m$.

Under some mild assumptions, one can show that for our studied problem with a linear feedback graph, UCB-LP achieves an improved regret upper bounded of $O(\log{T})$ compared to $O(n\log{T})$ of UCB1 and UCB-N. This result follows immediately from Proposition~10 of \cite{buccapatnam2017reward}.
Note that without the knowledge of $\delta_k$, UCB-LP achieves an improved regret upper bound that is similar to that of $\epsilon_t$-Greedy-LP. Although UCB-LP presented in Algorithm~\ref{alg:ucblp} requires the information about the time horizon $T$, this requirement can be relaxed using the techniques suggested in \cite{iucb,buccapatnam2017reward}.

Furthermore, leveraging the linear feedback graph in Fig.~\ref{fig:directedgraph}, we propose a further enhanced UCB algorithm by slightly modifying UCB-LP. We call this new algorithm UCB-LFG (UCB-Linear Feedback Graph) and present it in Algorithm~\ref{alg:ucblfg}.
The key difference is in Line~\ref{line:ucblfg-explore} (vs. Lines~\ref{line:ucblp-explore-start}-\ref{line:ucblp-explore-end} in Algorithm~\ref{alg:ucblp}).
Recall that in Algorithm~\ref{alg:ucblp}, the purpose of Lines~\ref{line:ucblp-explore-start}-\ref{line:ucblp-explore-end} is to explore arms in set $B_m$. Specifically, this is to ensure that by the end of stage $m$, each arm in set $B_m$ has at least $t_m$ samples. Because of the linear feedback graph, this exploration step can be achieved in a smaller number of rounds. Specifically, during stage $m$ we simply play arm $j^*(B_m)$ for $(t_m - t_{m-1})$ times, where $j^*(B_m)$ is the largest index of arms in set $B_m$. This is because each time when arm $j^*(B_m)$ is played, there will be a sample generated for every arm in set $B_m$, thanks to the special structure of the linear feedback graph. Following the regret analysis for UCB-LP in \cite{buccapatnam2017reward}, we can show that UCB-LFG achieves an improved regret upper bounded of $O(\log{T})$. Although UCB-LFG and UCB-LP have the same regret upper bound, UCB-LFG typically achieves a better empirical performance than UCB-LP. This can be observed from the simulation results in Section \ref{sec:sim_utility}.

\begin{figure*}[!t]
\centering
\begin{subfigure}[b]{0.31\linewidth}
	\centering
	\includegraphics[width=1\textwidth]{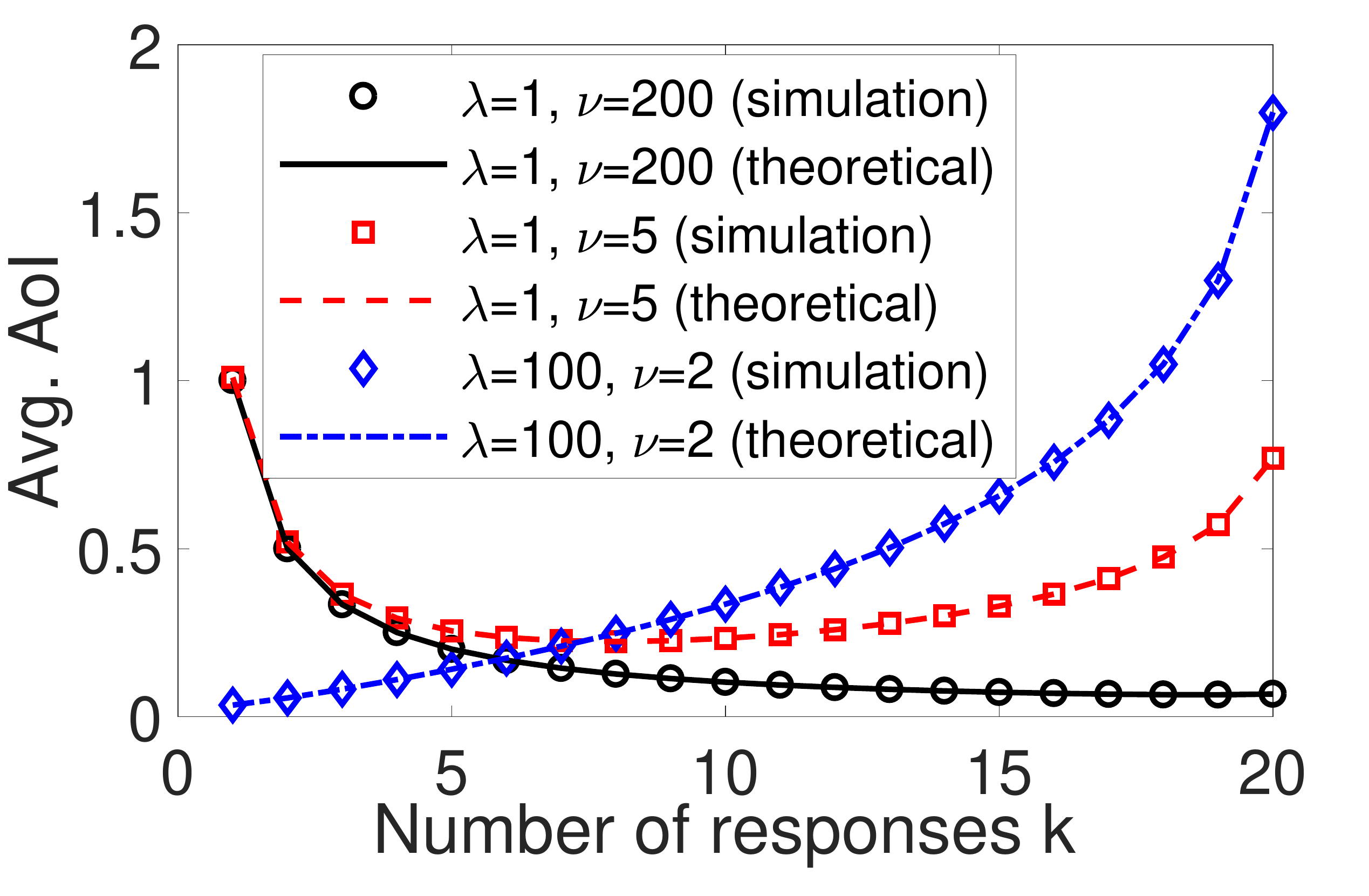}
	\caption{Exponential response time}
	\label{fig:exp}
\end{subfigure}
\quad
\begin{subfigure}[b]{0.31\linewidth}
	\centering
	\includegraphics[width=1\textwidth]{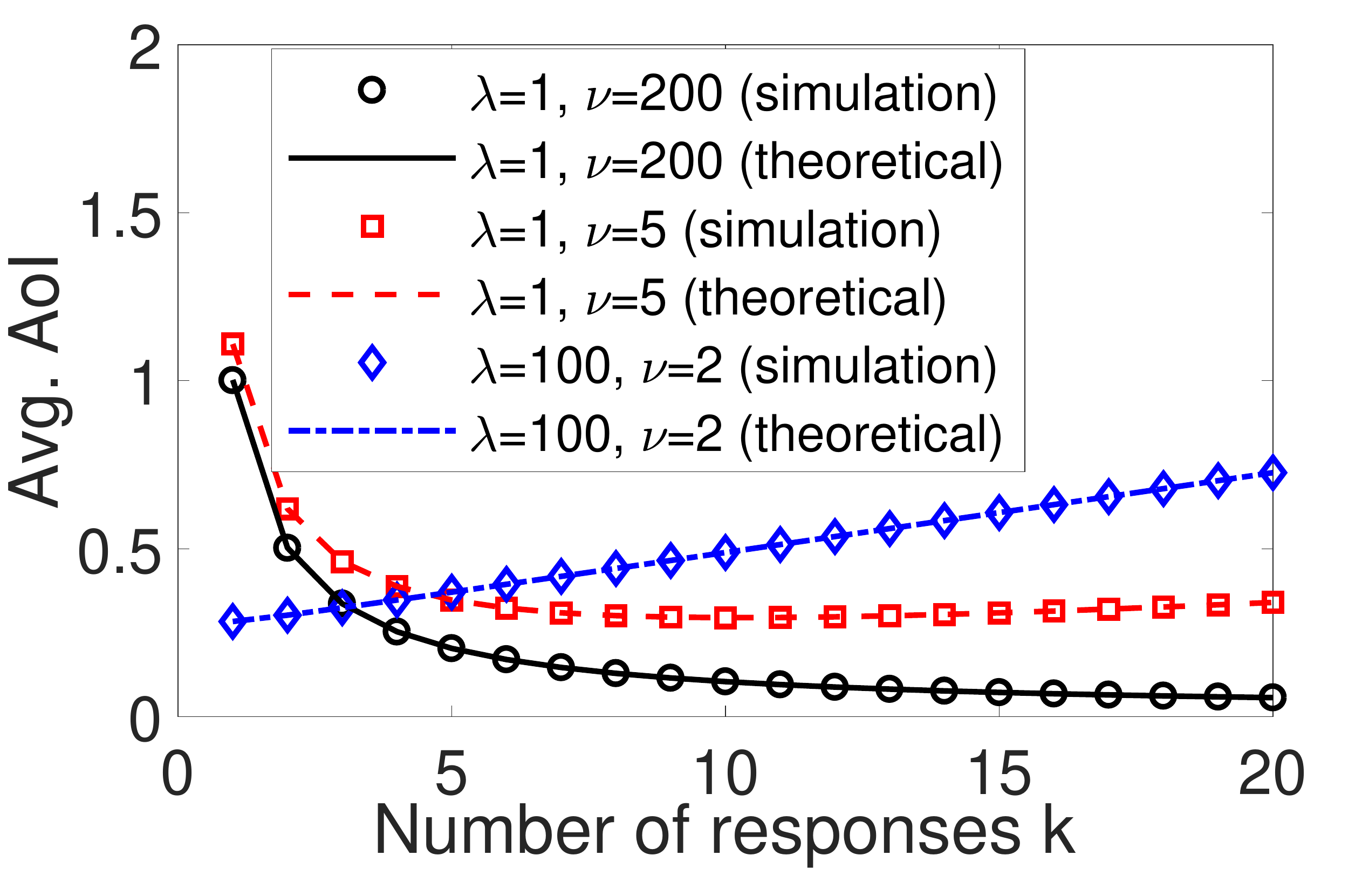}
	\caption{Uniform response time}
	\label{fig:uni}
\end{subfigure}
\quad
\begin{subfigure}[b]{0.31\linewidth}
	\centering
	\includegraphics[width=1\textwidth]{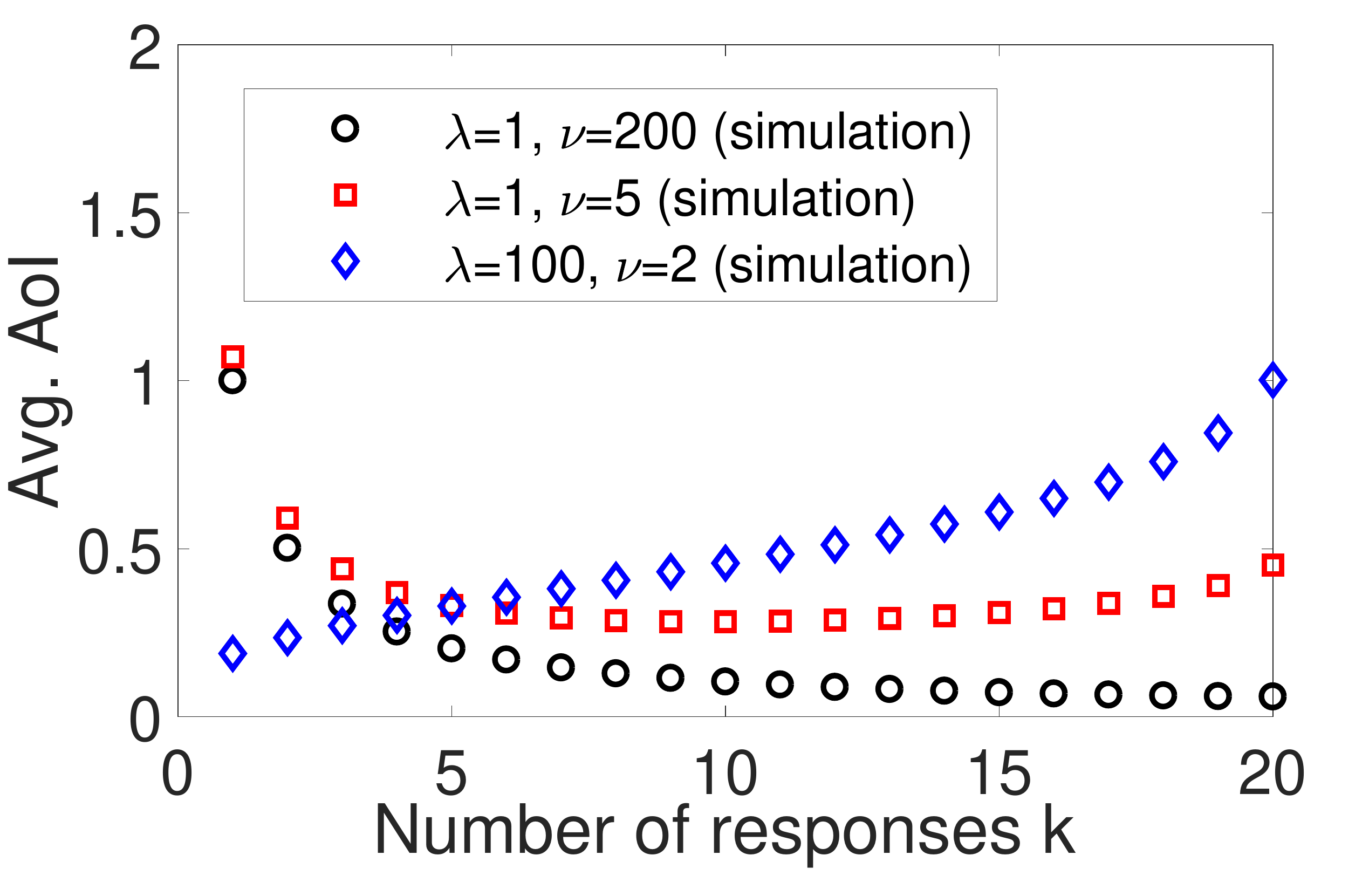}
	\caption{Gamma response time}
	\label{fig:gamma}
\end{subfigure}
\caption{Simulation results of average AoI vs. the number of responses $k$ for three different types of response time distributions}
\label{fig:sim}
\end{figure*}

\begin{figure*}[!t]
	\centering
	\begin{subfigure}[b]{0.31\linewidth}
		\includegraphics[width=1\textwidth]{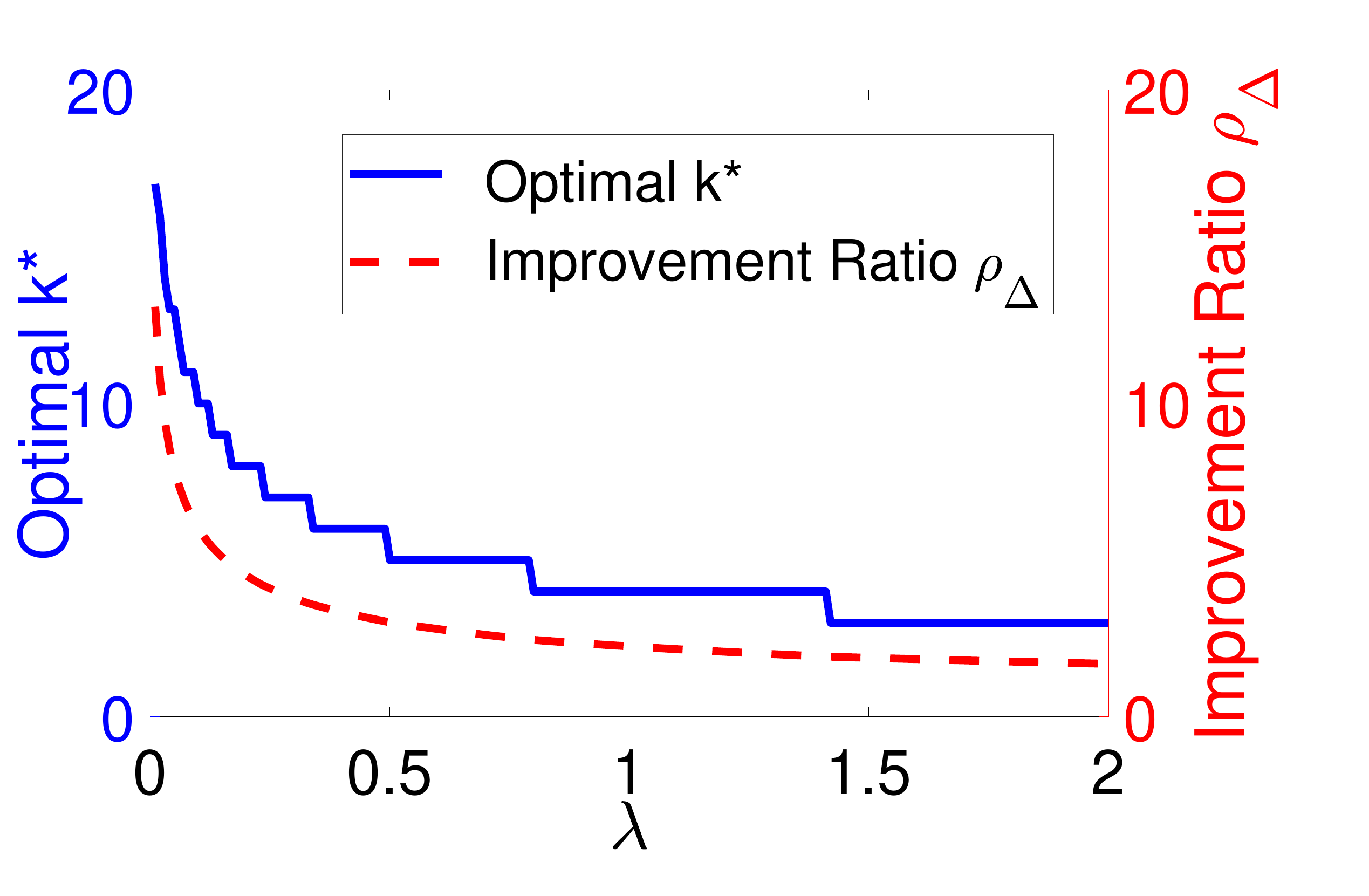}
		\caption{Impact of updating rate $\lambda$}
		\label{subfig:l}
	\end{subfigure}
	\quad
	\begin{subfigure}[b]{0.31\linewidth}
		\includegraphics[width=1\textwidth]{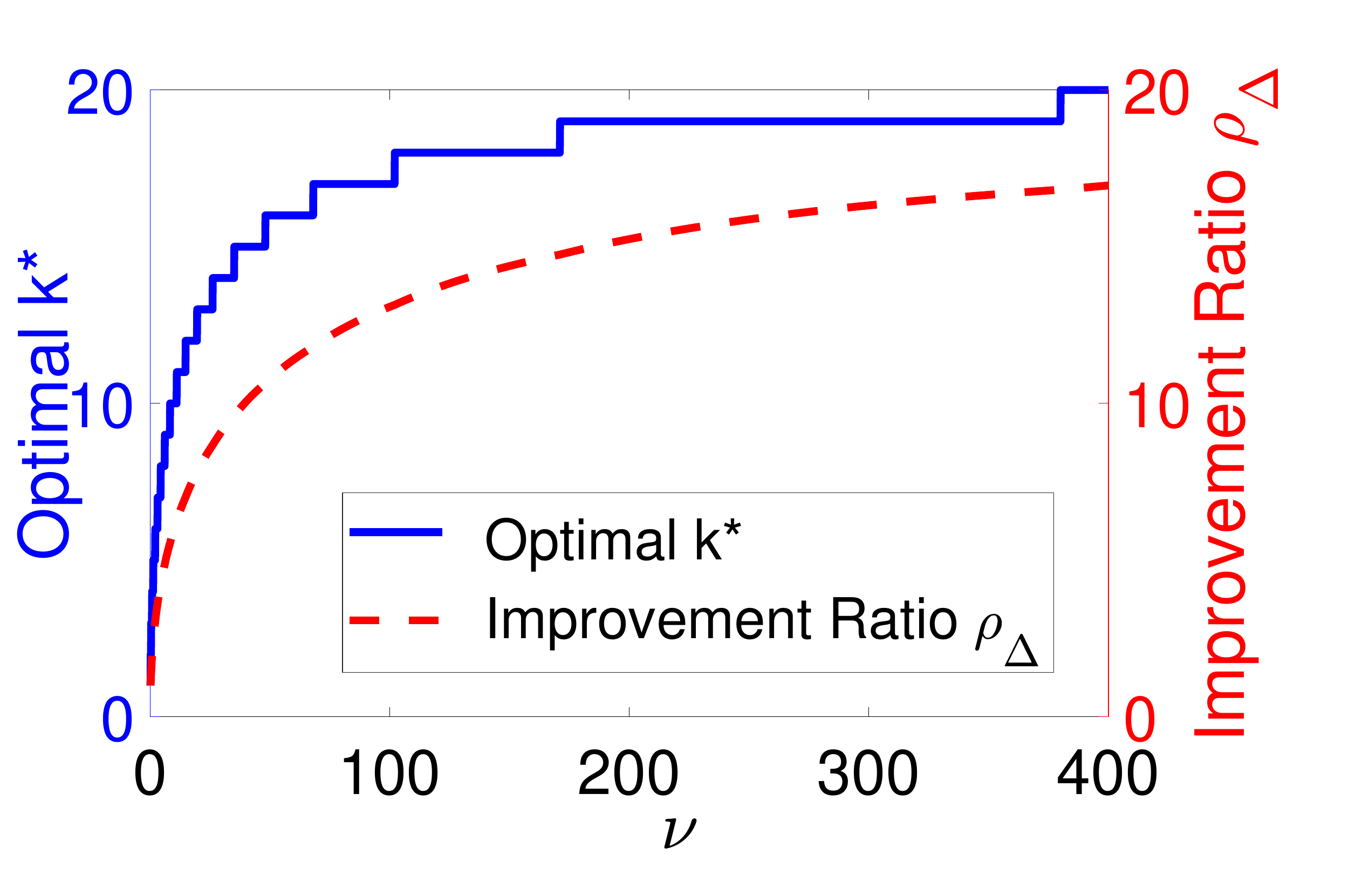}
		\caption{Impact of mean response time $1/\nu$}
		\label{subfig:m}
	\end{subfigure}
	\quad
	\begin{subfigure}[b]{0.31\linewidth}
		\includegraphics[width=1\textwidth]{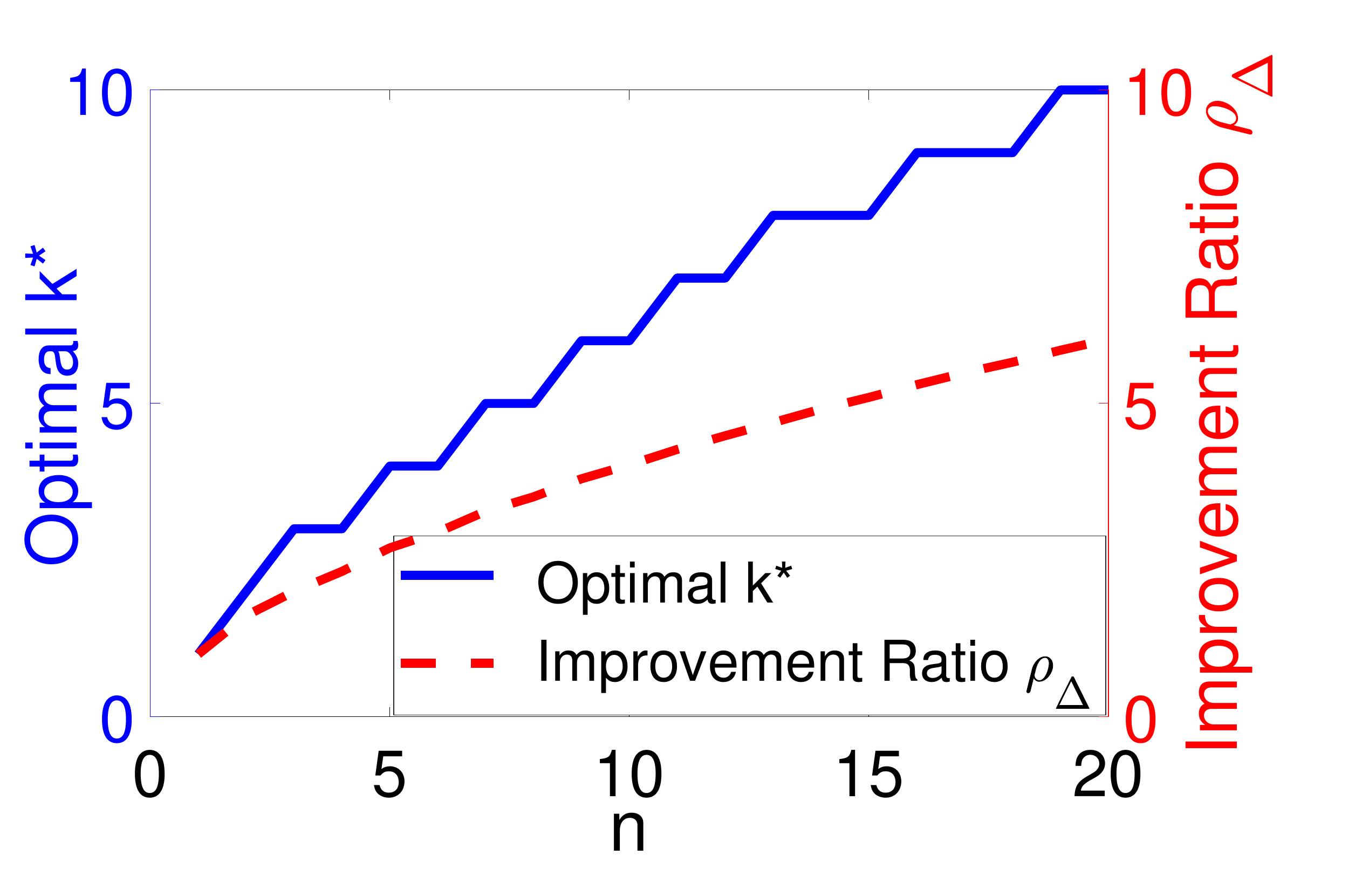}
		\caption{Impact of total number of servers $n$}
		\label{subfig:n}
	\end{subfigure}
	\caption{Impact of the system parameters on the optimal $k^*$ and the corresponding improvement ratio.
	We consider the exponential distribution for the response time.
	In (a), we fix $\nu = 1, n = 20$; in (b), we fix $\lambda = 1, n = 20$; in (c), we fix $\lambda = 1, \nu = 10$.}
	\label{fig:impacts}
\end{figure*}

\section{Numerical Results}\label{sec:sim}
In this section, we perform extensive simulations to elucidate our theoretical results. We present the simulation results for AoI minimization and AoI-based utility maximization in Sections~\ref{sec:sim_aoi} and \ref{sec:sim_utility}, respectively.

\subsection{Simulation Results for AoI Minimization}\label{sec:sim_aoi}
We first describe our simulation settings. We consider an information-update system with $n=20$ servers.
Throughout the simulations, the updating process at each server is assumed to be Poisson with rate $\lambda$ and is \emph{i.i.d.} across the servers. The user's request for the information is generated at time $s$, which is selected uniformly at random on interval $[0, 10^6/\lambda]$. This means that each server has a total of $10^6$ updates on average. 

Next, we evaluate the AoI performance through simulations for three types of response time distribution: 
\emph{exponential}, \emph{uniform}, and \emph{Gamma}.
First, we assume that the response time is exponentially distributed with mean $1/\nu$. We consider three representative setups: (i) $\lambda=1, \nu = 200$; (ii) $\lambda=1, \nu=5$; (iii) $\lambda=100, \nu = 2$.
Fig.~\ref{fig:exp} shows how the average AoI changes as the number of responses $k$ varies, where each point represents an average of $10^3$ simulation runs.
We also include plots of our theoretical results (i.e., Eq.~\eqref{eq:e_aoi_formula}) for comparison.
A crucial observation from Fig.~\ref{fig:exp} is that the simulation results match our theoretical results perfectly. 
In addition, we observe three different behaviors of the average AoI performance:
a) If the inter-update time is much larger than the response time (e.g., $\lambda = 1$, $\nu = 200$),
then the average AoI decreases as $k$ increases, and thus, it is worth waiting for all the responses so as to achieve a smaller average AoI. 
b) In contrast, if the inter-update time is much smaller than the response time (e.g., $\lambda = 100$, $\nu = 2$), 
then the average AoI increases as $k$ increases, and thus, it is not beneficial to wait for more than one response. 
c) When the inter-update time is comparable to the response time (e.g., $\lambda=1$, $\nu=5$),
then as $k$ increases, the AoI would first decrease and then increase. 
In this setup, when $k$ is small, the freshness of the data at the servers dominates, and thus, waiting for more responses helps reduce the average AoI. 
On the other hand, when $k$ becomes large, the total waiting time becomes dominant, and thus, the average AoI increases as $k$ further increases.

In Section~\ref{sec:extensions}, we discussed the extension of our theoretical results to the case of uniformly distributed response time.
Hence, we also perform simulations for the response time uniformly distributed on $[\frac{1}{2\nu}, \frac{3}{2\nu}]$ with mean $1/\nu$. 
Fig.~\ref{fig:uni} presents the average AoI as the number of responses $k$ changes. In this scenario, the simulation results also
match our theoretical results (i.e., Eq.~\eqref{eq:e_aoi_uniform}). 
Also, we observe a very similar phenomenon to that in Fig.~\ref{fig:exp} on how the average AoI varies as $k$ increases in three 
different simulation setups. 

In addition, Fig.~\ref{fig:gamma} presents the simulation results for the response time with Gamma distribution, 
which can be used to model the response time in relay networks\cite{najm16}. 
Specifically, we consider a special class of the Gamma($r,\theta$) distribution that is the sum 
of $r$ \emph{i.i.d.} exponential random variables with mean $\theta$  (which is also called the Erlang distribution).
Then, the mean response time $1/\nu$ is equal to $r\theta$. We fix $r=5$ in the simulations.
In this case, although we do not have any analytical results, the observations are similar to that under the exponential and uniform distributions. 

Finally, we investigate the impact of the system parameters (the updating rate, the mean response time, 
and the total number of servers) on the optimal number of responses $k^*$ and the \emph{AoI improvement ratio}, 
defined as $\rho_{\Delta} \triangleq \mathds{E}[\Delta(1)]/\mathds{E}[\Delta(k^*)]$.
The AoI improvement ratio captures the gain in the AoI reduction under the optimal scheme 
compared to a naive scheme of waiting for the first response only. 

Fig.~\ref{subfig:l} shows the impact of the updating rate $\lambda$. 
We observe that the optimal number of responses $k^*$ decreases as $\lambda$ increases. 
This is because when the updating rate is large, the AoI diversity at the servers is small. 
In this case, waiting for more responses is unlikely to receive a response with much fresher 
information. Therefore, 
the optimal scheme will simply be a naive scheme that waits for the first response only, when the updating rate is relatively large (e.g., $\lambda=2$). 
%
%
Fig.~\ref{subfig:m} shows the impact of the mean response time $1/\nu$. 
We observe that the optimal number of responses $k^*$ increases as $\nu$ increases.
This is because when $\nu$ is large (i.e., when the mean response time is small),
the cost of waiting for additional responses becomes marginal, and thus, waiting for more 
responses is likely to lead to the reception of a response with fresher information.
%
Fig.~\ref{subfig:n} shows the impact of the total number of servers $n$.
We observe that both the optimal number of responses $k^*$ and
the improvement ratio increase with $n$.
This is because an increased number of servers leads to more diversity gains
both in the AoI at the servers and in the response time.
As we discussed at the end of Section~\ref{sec:aoi-opt}, the optimal solution $k^*$ scales with $O(\sqrt{n})$ as the number of servers $n$ becomes large.


\begin{figure*}[!t]
\centering
\begin{subfigure}[b]{0.31\linewidth}
	\centering
	\includegraphics[width=1\textwidth]{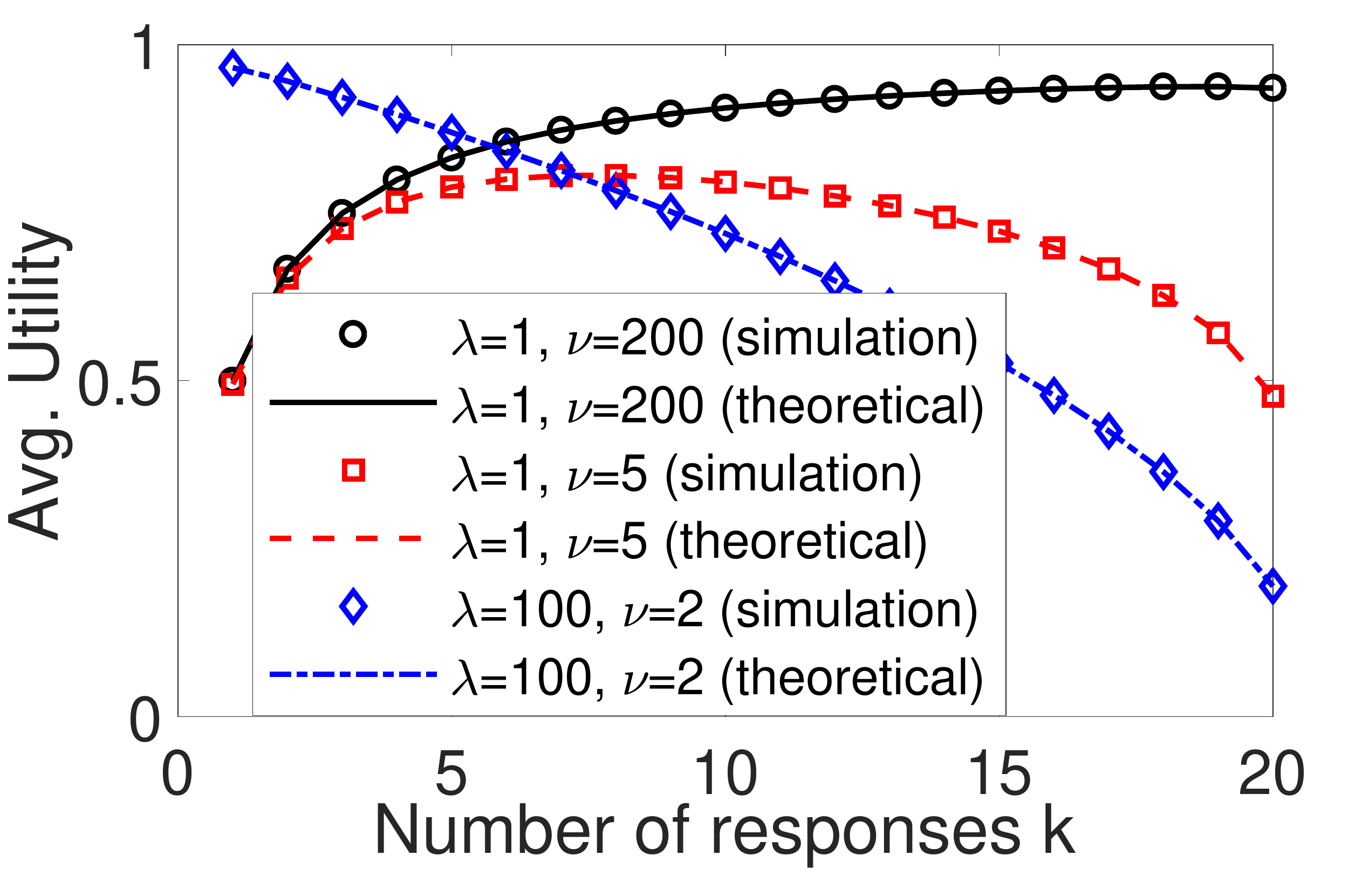}
	\caption{Exponential response time}
	\label{fig:rd_exp}
\end{subfigure}
\quad
\begin{subfigure}[b]{0.31\linewidth}
	\centering
	\includegraphics[width=1\textwidth]{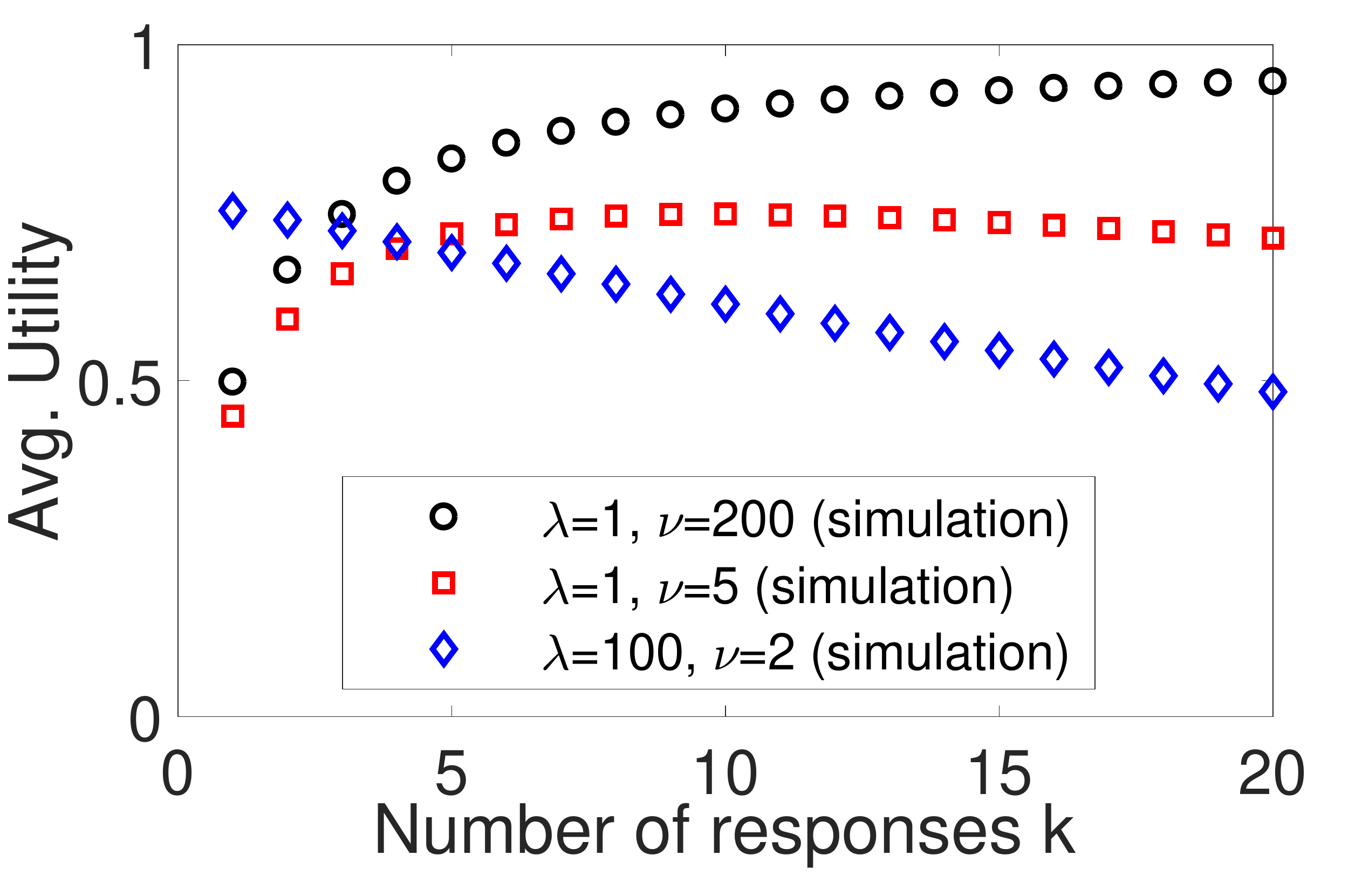}
	\caption{Uniform response time}
	\label{fig:rd_uni}
\end{subfigure}
\quad
\begin{subfigure}[b]{0.31\linewidth}
	\centering
	\includegraphics[width=1\textwidth]{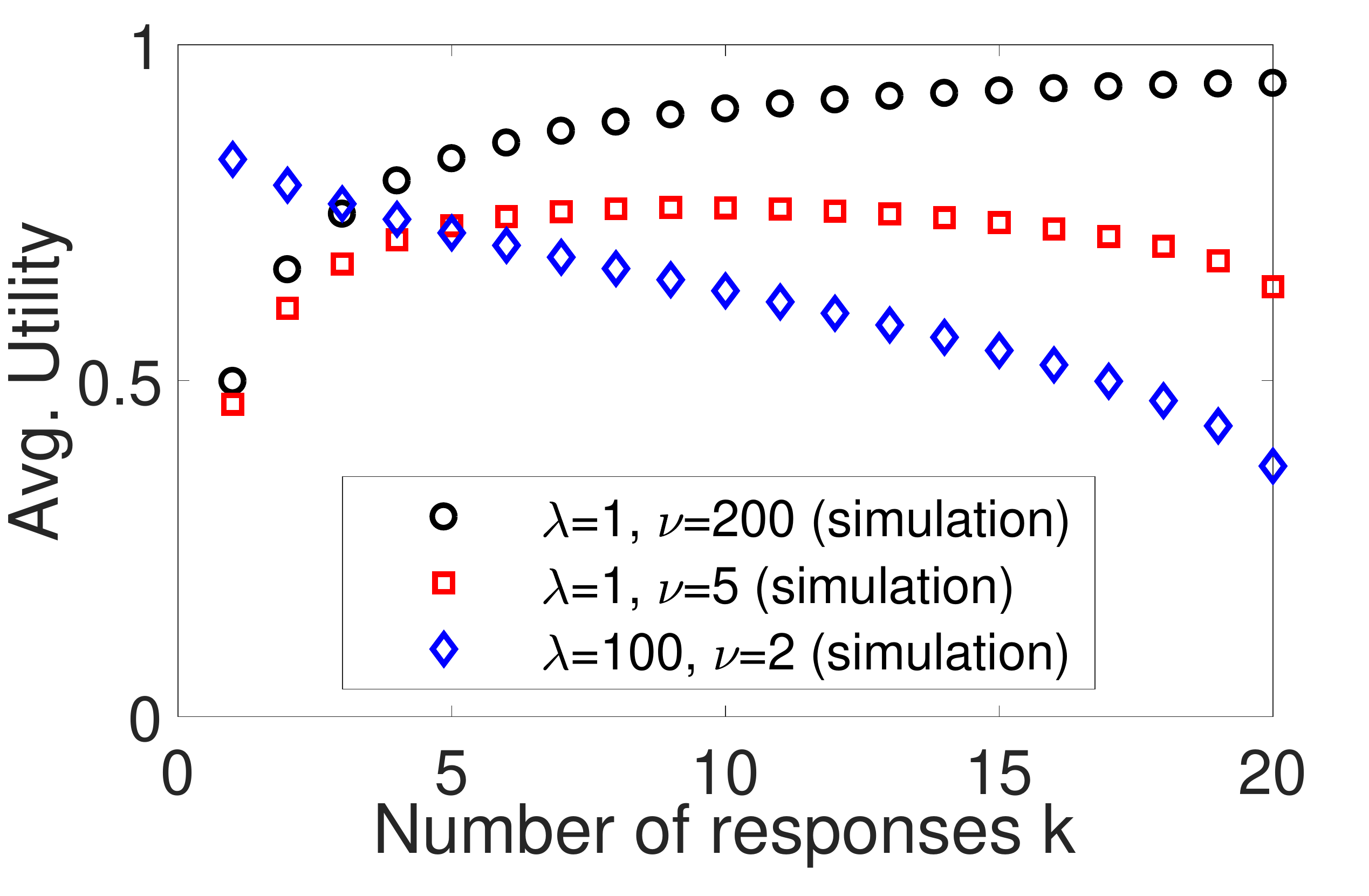}
	\caption{Gamma response time}
	\label{fig:rd_gamma}
\end{subfigure}
\caption{Simulation results of average utility vs. the number of responses $k$ for three different types of response time distributions}
\label{fig:sim_rd}
\end{figure*}

\begin{figure*}[!t]
	\centering
	\begin{subfigure}[b]{0.31\linewidth}
		\includegraphics[width=1\textwidth]{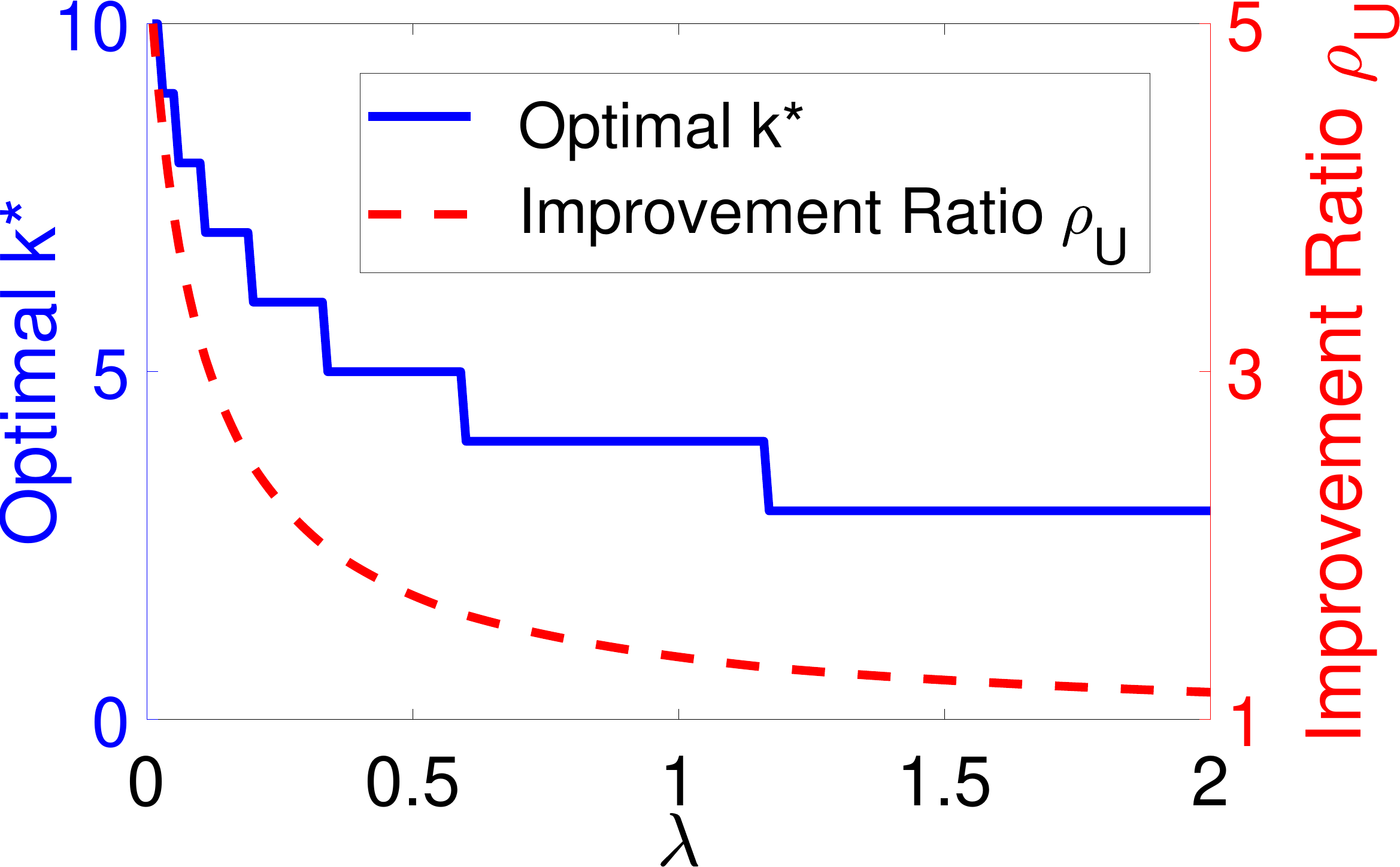}
		\caption{Impact of updating rate $\lambda$}
		\label{subfig:rd_l}
	\end{subfigure}
	\quad
	\begin{subfigure}[b]{0.31\linewidth}
		\includegraphics[width=1\textwidth]{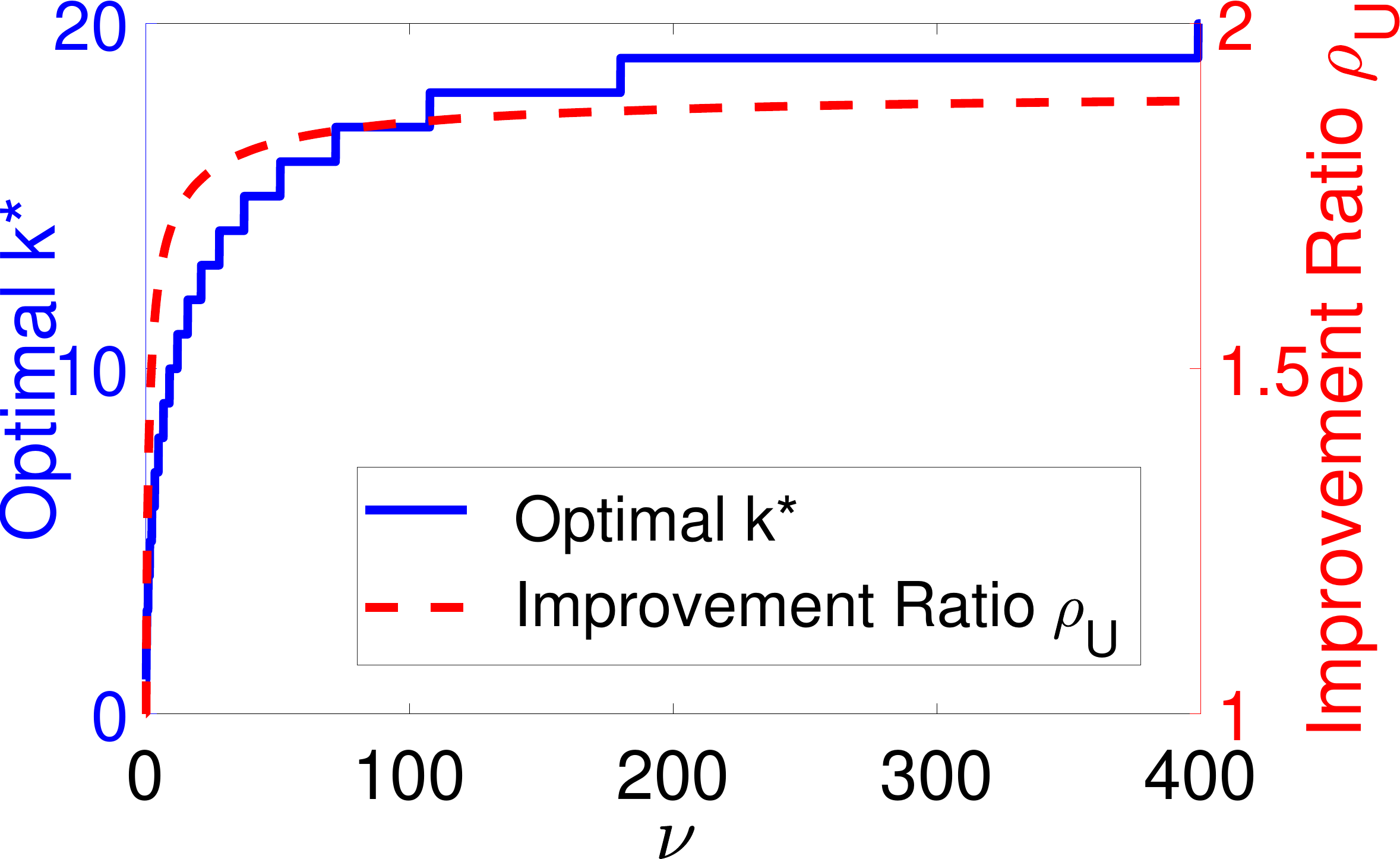}
		\caption{Impact of mean response time $1/\nu$}
		\label{subfig:rd_m}
	\end{subfigure}
	\quad
	\begin{subfigure}[b]{0.31\linewidth}
		\includegraphics[width=1\textwidth]{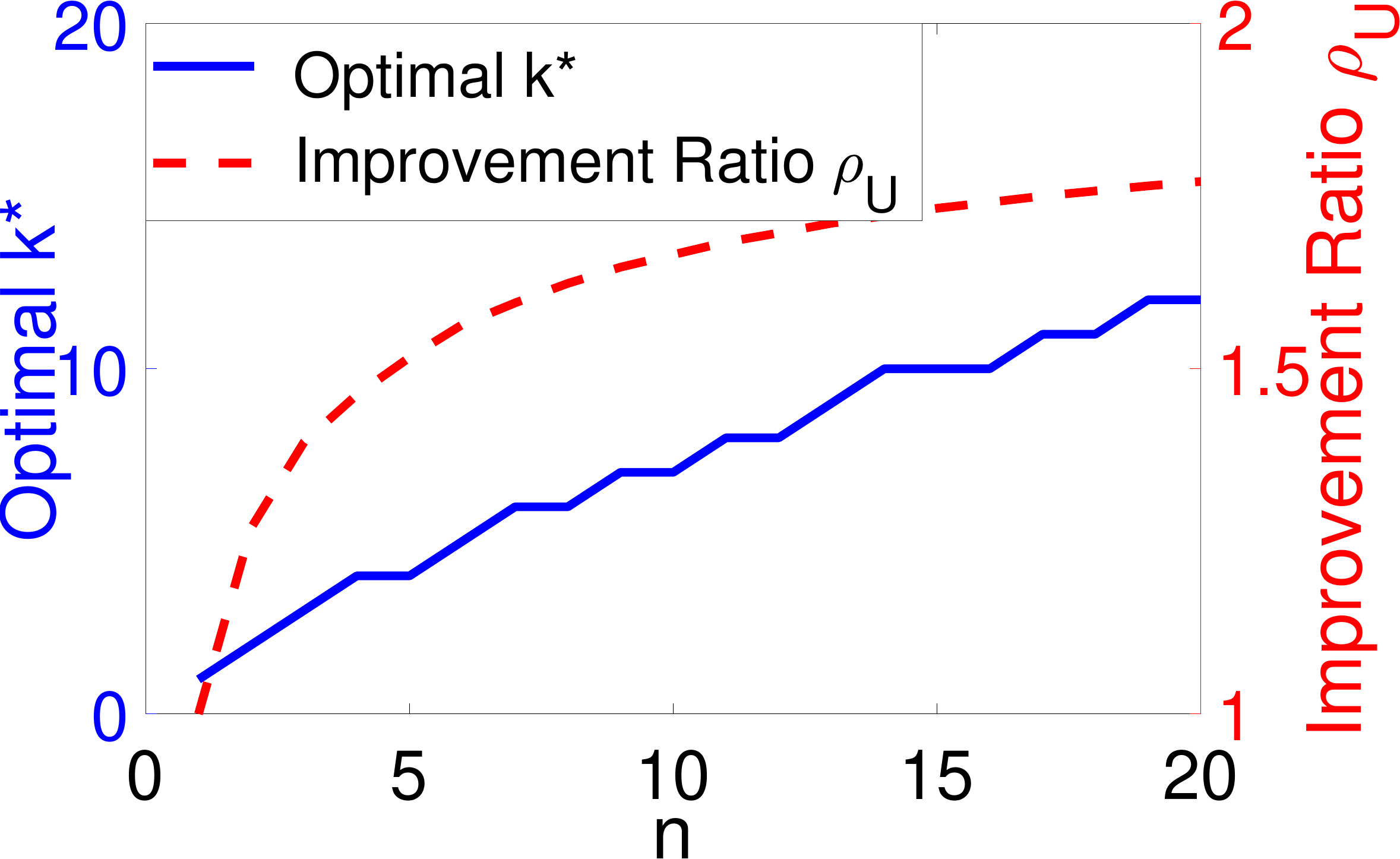}
		\caption{Impact of total number of servers $n$}
		\label{subfig:rd_n}
	\end{subfigure}
	\caption{Impact of the system parameters on the optimal $k^*$ for utility maximization and the corresponding improvement ratio.
	We consider the exponential distribution for the response time.
	In (a), we fix $\nu = 1, n = 20$; in (b), we fix $\lambda = 1, n = 20$; in (c), we fix $\lambda = 1, \nu = 10$.}
	\label{fig:rd_impacts}
\end{figure*}

\begin{figure*}[!t]
		\centering
		\begin{subfigure}[b]{0.31\linewidth}
			\includegraphics[width=1\textwidth]{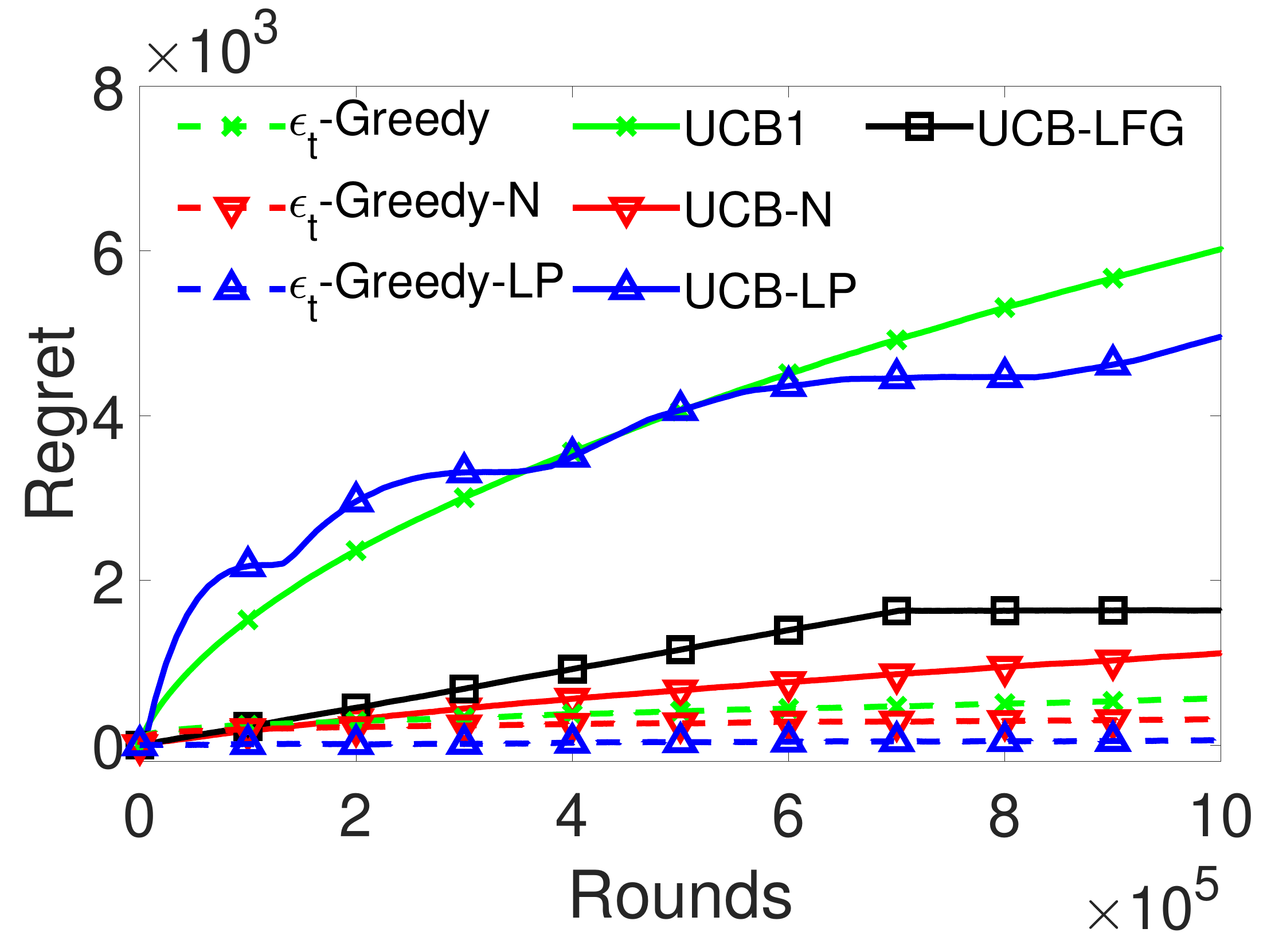}
			\caption{Setup (i): $\lambda=1, \nu=200$}
			\label{fig:regret-cmp-1}
		\end{subfigure}
		\quad
		\begin{subfigure}[b]{0.31\linewidth}
			\includegraphics[width=1\textwidth]{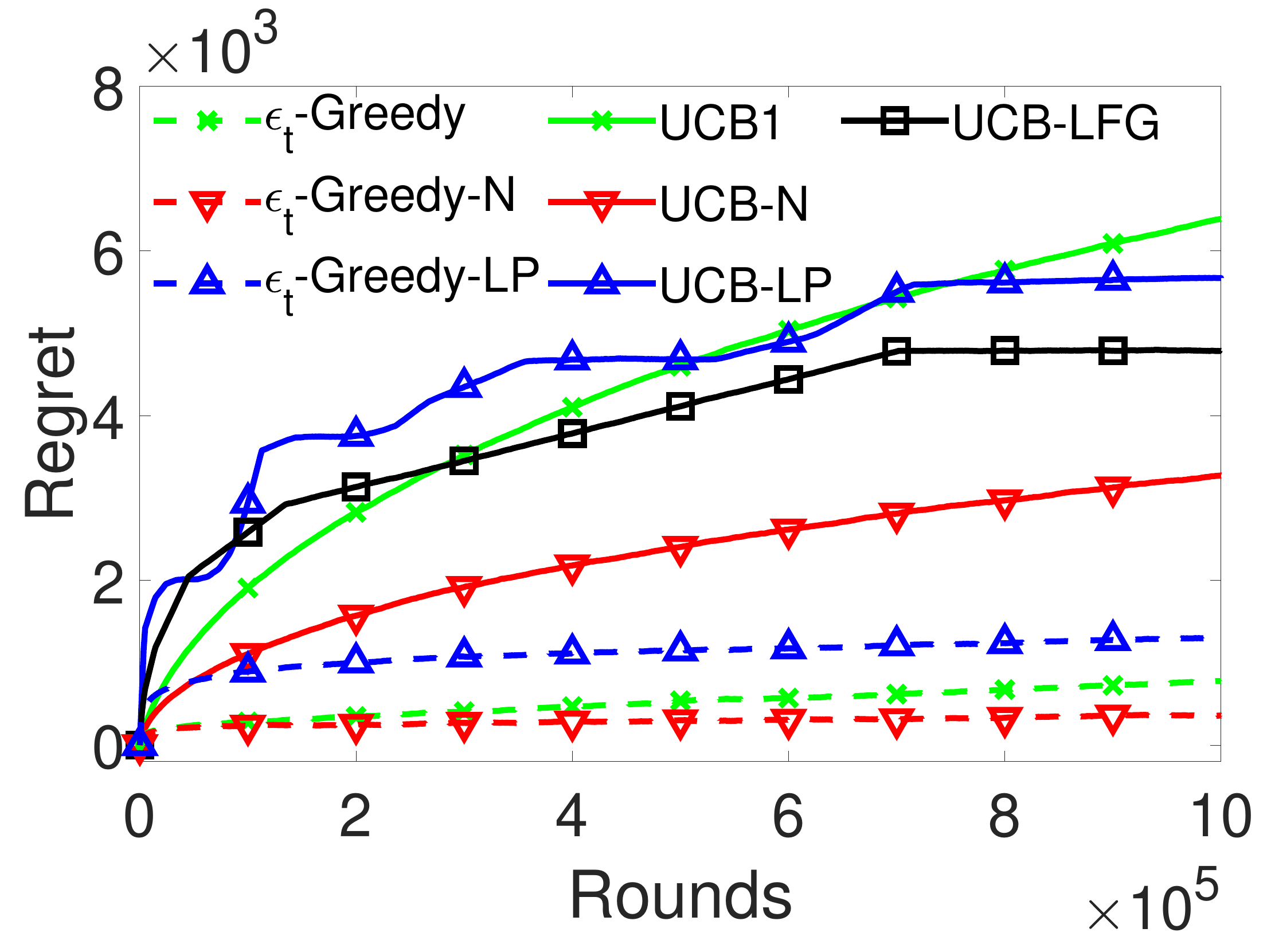}
			\caption{Setup (ii): $\lambda=1, \nu=5$}
			\label{fig:regret-cmp-2}
		\end{subfigure}
		\quad
		\begin{subfigure}[b]{0.31\linewidth}
			\includegraphics[width=1\textwidth]{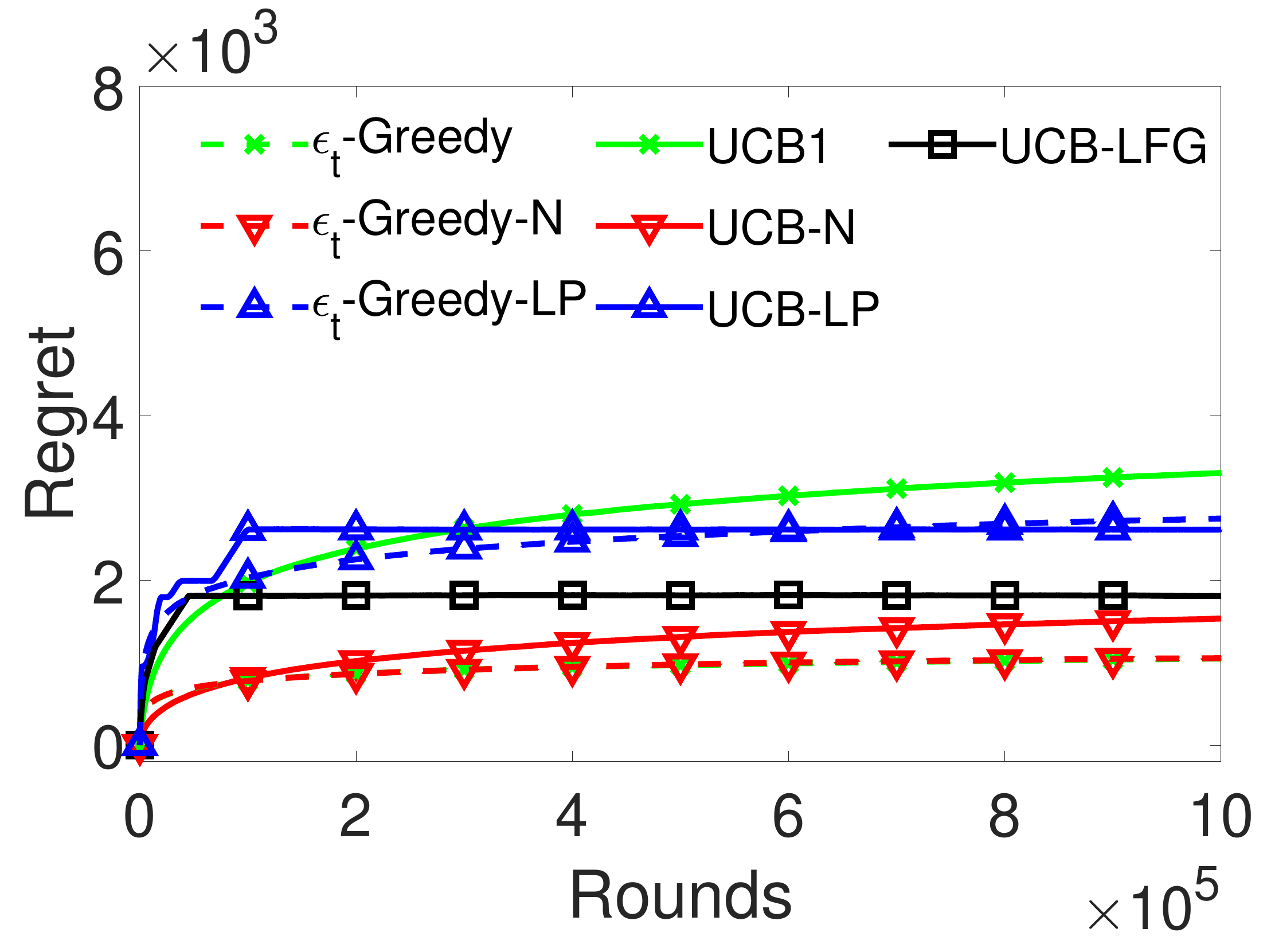}
			\caption{Setup (iii): $\lambda=100, \nu=2$}
			\label{fig:regret-cmp-3}
		\end{subfigure}
		\caption{Comparison of regret performance of various learning algorithms in three representative setups}
		\label{fig:regret-cmp}
\end{figure*}

\subsection{Simulation Results for AoI-based Utility Maximization}\label{sec:sim_utility}


In this subsection, we consider the maximization of an AoI-based exponential utility function (Eq.~\eqref{eq:utility_exp} with $a=1$) and present simulation results for the utility performance under the same settings as in Section \ref{sec:sim_aoi}. We first evaluate the utility performance in the setting with known system parameters. Then, we evaluate various learning algorithms described in Section \ref{sec:mab_alg}, when the system parameters are unknown.

\subsubsection{Utility Maximization with Known Parameters} 

The setups we consider are exactly the same as those in Section~\ref{sec:sim_aoi}, except that we now focus on the utility performance instead of the AoI performance.
In Fig.~\ref{fig:sim_rd}, we present the simulation results for the average utility performance with a varying number of responses $k$ in three representative setups. The observations are also similar, except that the utility has an opposite trend compared to the AoI. This is because the utility is a non-increasing function of the AoI.

Similarly, we also investigate the impact of the system parameters on the optimal number of responses $k^*$ (with respect to utility maximization) and the \emph{utility improvement ratio}, 
defined as $\rho_{U} \triangleq \mathds{E}[U(\Delta(k^*))]/\mathds{E}[U(\Delta(1))]$.
The utility improvement ratio captures the gain in the utility improvement under the optimal scheme 
compared to a naive scheme of waiting for the first response only. 
We present the results in Fig.~\ref{fig:rd_impacts}, from which we can make similar observations to those from Fig.~\ref{fig:impacts} for AoI minimization.

\subsubsection{Utility Maximization with Unknown Parameters} 

Next, we consider a more realistic scenario where the system parameters (i.e., the updating rate and the mean response time) are unknown to the user. Given that the overall behaviors are similar for different types of response time distributions (see Fig.~\ref{fig:sim_rd}), in the following evaluations we will focus on the case where the update process is Poisson with rate $\lambda$ and the response time is exponentially distributed with mean $1/\nu$. 

As in Section~\ref{sec:sim_aoi}, we assume $n=20$ and consider three representative setups: (i) $\lambda=1, \nu = 200$; (ii) $\lambda=1, \nu=5$; (iii) $\lambda=100, \nu = 2$.
We evaluate the regret performance of two classes of learning algorithms we introduced in Section~\ref{sec:mab_alg}: the Greedy algorithms (i.e., $\epsilon_t$-Greedy, $\epsilon_t$-Greedy-N, and $\epsilon_t$-Greedy-LP) and the UCB algorithms\footnote{We do not include the results for UCB-Improved as it performs much worse than the other algorithms in the setups we consider.} (i.e., UCB1, UCB-N, UCB-Improved, UCB-LP, and UCB-LFG). For Greedy algorithms, we use $d = 0.05$ and $c = 1$ in all the three setups.
In Fig.~\ref{fig:regret-cmp}, we plot the evolution of cumulative regret of the considered algorithms over $10^6$ rounds for each of the above three setups. 
The results represent an average of 10 simulation runs. 
From the simulation results in Fig.~\ref{fig:regret-cmp}, we can observe the following. 

First, algorithms that take advantage of side observations generally outperform their counterparts that do not use side observations. That is, $\epsilon_t$-Greedy-N and UCB-N outperform $\epsilon_t$-Greedy and UCB1, respectively. This is because additional samples from side observations can help accelerate the learning process.

Second, although graph-aware algorithms can achieve improved regret upper bounds, their empirical performances may or may not be better than that of their graph-agnostic counterparts. That is, $\epsilon_t$-Greedy-LP and UCB-LP may or may not be better than $\epsilon_t$-Greedy-N and UCB-N, respectively. Consider the Greedy algorithms for example. In Setup (i), $\epsilon_t$-Greedy-LP slightly outperforms $\epsilon_t$-Greedy-N. This is because in the phase of exploration, $\epsilon_t$-Greedy-LP always chooses arm $n$, which happens to be the best arm. However, in Setup (ii), $\epsilon_t$-Greedy-LP performs worse than $\epsilon_t$-Greedy-N. This is because arm $n$ is no longer the best arm, and in fact, it can be much worse than the optimal arm. This phenomenon is exacerbated in Setup (iii), where arm $n$ is the worst arm. Among all the considered UCB algorithms, UCB-N has the best empirical performance. This is because UCB-LP and UCB-LFG are modified from UCB-Improved, which is an ``arm-elimination" algorithm and is very different from UCB1, from which UCB-N is modified. Although UCB-Improved has a better regret upper bound with a smaller constant factor, it has a much worse empirical performance than UCB1 in the setups we consider. Therefore, it is not surprising that UCB-N has a better empirical performance than UCB-LP and UCB-LFG.

Third, UCB-LFG typically outperforms UCB-LP. This is expected because UCB-LFG is a further enhanced version of UCB-LP. Specifically, UCB-LFG explicitly exploits the linear structure of the feedback graph and can accelerate the learning process by reducing the number of rounds for exploration.

Finally, $\epsilon_t$-Greedy-N seems to be quite robust and has a very good empirical performance in all the setups we consider.

\section{Conclusion}\label{sec:conclusion}
In this paper, we introduced a new Pull model for studying the problems of AoI minimization and AoI-based utility maximization under the replication schemes.
Assuming Poisson updating process and exponentially distributed response time, we derived the closed-form expression
of the expected AoI at the user's side and provided a formula
for computing the optimal solution. 
We also derived a set of similar theoretical results for the utility maximization problem. 
Furthermore, we considered a more realistic scenario where the user has no prior knowledge of the system parameters.
In this setting, we reformulated the utility maximization problem as a stochastic MAB problem with side observations. Leveraging the special linear structure of the feedback graph associated with side observations, we introduced several learning algorithms, which outperform those basic algorithms that are agnostic about such properties.
Not only did our work reveal a novel tradeoff between different levels of information
freshness and different response times across the servers, but
we also demonstrated the power of waiting for more than one
response in minimizing the AoI as well as in maximizing the utility at the user's side.

\bibliographystyle{IEEEtran}
\bibliography{ref,aoi}

\appendix 

\subsection{Proof of Theorem \ref{thm:e_utility}} \label{app:e_utility}
\begin{proof}
Note that using Eq.~\eqref{eq:e_utility}, the expected utility can be computed based on the probability density function of the AoI. For the exponential utility function \eqref{eq:utility_exp}, however, we have the following more intuitive way of computing the expected utility.

To begin with, we rewrite the expected utility as follows:
\begin{equation}
\label{eq:e_utility_new}
\begin{split}
\bE[U(\Delta(k))] &=\bE \left[e^{-a\Delta(k)}\right] \\
&=\bE \left[e^{-aR_{(k)} - a\min_{i \in \mK} \Delta_i(s)}\right] \\
& = \bE\left[e^{-aR_{(k)} }\right] \cdot \bE \left[e^{ - a\min_{i \in \mK} \Delta_i(s)}\right],
\end{split}
\end{equation}
where the first equality is from Eq.~\eqref{eq:utility_exp}, the second equality is from Eq.~\eqref{eq:aoik_l}, and the last equality is due to that $R_{(k)}$ and $\min_{i \in \mK} \Delta_i (s)$ are independent.

Then, we want to derive the expression for each of the two terms in the last line of Eq.~\eqref{eq:e_utility_new}.

First, we want to show $\bE[e^{-aR_{(k)}}] = \prod_{j=1}^k \frac{(n+1-j)\nu}{(n+1-j)\nu +a}$. Note that for an exponential random variable $X$ with mean $1/\alpha$, variable $aX$ is also an exponential random variable but with mean $a/\alpha$, and it is easy to show the following:
\begin{equation}
\label{eq:exp_exp}
\bE[e^{-aX}] = \frac{\alpha}{\alpha+a},
\end{equation}
Recall from the proof of Theorem~\ref{thm:e_aoi} that $R_{(j)} - R_{(j-1)}$ is an exponential random variable with mean $\frac{1}{(n+1-j)\nu}$ for any $j \in \mN$ and that the exponential random variables $(R_{(j)} - R_{(j-1)})$'s are all independent~\cite{firstcourse}. Then, we can derive the following:
\begin{equation}
\begin{split}
\label{eq:e_expR_k}
\bE \left[e^{-aR_{(k)}}\right] &= \bE \left[e^{-a\sum_{j=1}^{k} (R_{(j)} - R_{(j-1)})}\right] \\
&=\bE \left[\prod_{j=1}^{k}e^{-a(R_{(j)} - R_{(j-1)})}\right] \\
& =\prod_{j=1}^{k}\bE \left[e^{-a(R_{(j)} - R_{(j-1)})}\right] \\
&= \prod_{j=1}^k \frac{(n+1-j)\nu}{(n+1-j)\nu +a},
\end{split}
\end{equation}
where the last equality is from Eq.~\eqref{eq:exp_exp}. 

Next, we want to show $\bE \left[e^{ -a \min_{i \in \mK} \Delta_i(s)}\right] = \frac{k\lambda}{k\lambda+a}$. Recall that $\min_{i \in \mK} \Delta_i(s)$ is an exponential random variable with mean $\frac{1}{k\lambda}$. Then, it is straightforward to see $\bE \left[e^{ -a \min_{i \in \mK} \Delta_i(s)}\right] = \frac{k\lambda}{k\lambda+a}$ due to Eq.~\eqref{eq:exp_exp}.

Combining the above results, we complete the proof.
\end{proof}
\subsection{Proof of Theorem \ref{thm:utility_opt}} \label{app:utility_opt}
\begin{proof}
	We first define $r(k)$ as the ratio of the expected utility between the $(n, k+1)$ and $(n,k)$ replication schemes, i.e., $r(k) \triangleq \bE[U(\Delta(k+1))]/\bE[U(\Delta(k))]$ for any $k \in \{1,2,\dots,n-1\}$.
	From Eq.~\eqref{eq:e_utility_formula}, we have the following:
	\begin{equation}\label{eq:ratio}
	\begin{split}
	\\r(k) &= \frac{(k+1)(k\lambda+a)}{k\left((k+1)\lambda+a\right)} \cdot \frac{(n-k)\nu}{(n-k)\nu+a} \\
	&=\left(1+\frac{a}{\lambda k^2+(\lambda+a)k}\right)\cdot \frac{\nu}{\frac{a}{n-k}+\nu},\\
	\end{split}
	\end{equation}
	for any $k \in \{1,2,\dots,n-1\}$.
	It is easy to see that $r(k)$ is a monotonically decreasing function of $k$.
	
	We now extend the domain of $r(k)$ to the set of positive real numbers and want to find $k^{\prime}$
	such that $r(k^{\prime}) = 1$. With some standard calculations and dropping the negative solution, 
	we derive the following:
	\begin{equation}
	k^{\prime} = \frac{2\nu n}{\sqrt{(\lambda + \nu+a)^2+4\lambda\nu n}+\lambda+\nu+a}.
	\end{equation}

	Next, we discuss two cases: (i) $k^{\prime}>n-1$ and (ii) $0 < k^{\prime} \le n-1$.
	
	In Case (i), we have $k^{\prime}>n-1$. This implies that $r(k)=\bE[U(\Delta(k+1))]/\bE[U(\Delta(k))]>1$ for all $k \in \{1,2,\dots,n-1\}$, as $r(k)$ is monotonically decreasing. Hence, the expected utility, $\bE[U(\Delta(k))]$, is a monotonically increasing function for $k \in \{1,2,\dots,n\}$.
	Therefore, $k^*=n$ must be the optimal solution.
	
	In Case (ii), we have $0< k^{\prime} \le n-1$. We consider two subcases: $k^{\prime}$ is an integer in $\{1,2,\dots,n-1\}$ and $k^{\prime}$ is not an integer.
	
	If $k^{\prime}$ is an integer in $\{1,2,\dots,n-1\}$, we have $r(k)=\bE[U(\Delta(k+1))]/\bE[U(\Delta(k))] \ge 1$ for $k \in \{1,2,\dots,k^{\prime}\}$ and $r(k)=\bE[U(\Delta(k+1))]/\bE[U(\Delta(k))]<1$ 
	for $k\in\{k^{\prime}+1,\dots,n\}$, as $r(k)$ is monotonically decreasing.
	Hence, the expected utility, $\bE[U(\Delta(k))]$, is first increasing (for $k \in \{1,2,\dots,k^{\prime}\}$) and then decreasing (for $k\in\{k^{\prime}+1,\dots,n\}$).
	Therefore, there are two optimal solutions: $k^*=k^{\prime}$ and $k^*=k^{\prime}+1$ since $\bE[U(\Delta(k^\prime+1))] = \bE[U(\Delta(k^\prime))]$ (due to $r(k^{\prime})=1$).
	
	If $k^{\prime}$ is not an integer, we have $r(k)=\bE[U(\Delta(k+1))]/\bE[U(\Delta(k))]>1$ for $k \in \{1,2,\dots,\lfloor k^{\prime} \rfloor\}$ and $r(k)=\bE[U(\Delta(k+1))]/\bE[U(\Delta(k))]<1$ 
	for $k\in\{\lceil k^{\prime} \rceil,\dots,n\}$, as $r(k)$ is monotonically decreasing.
	Hence, the expected reward $\mu(k)$ is first increasing (for $k \in \{1,2,\dots,\lfloor k^{\prime} \rfloor, \lceil k^{\prime} \rceil \}$) and then decreasing (for $k\in\{\lceil k^{\prime} \rceil,\dots,n\}$).
	Therefore, $k^* = \lceil k^{\prime} \rceil$ must be the optimal solution.
	
	Combining two subcases, we have $k^* = \lceil k^{\prime} \rceil$ in Case (ii). Then, combining Cases (i) and (ii), we have 
	$k^* = \min \{\lceil k^{\prime} \rceil, n\} = \min \{ \lceil \frac{2\nu n}{\sqrt{(\lambda + \nu+a)^2+4\lambda\nu n}+\lambda+\nu+a} \rceil, n \}$.
\end{proof}
\subsection{Proof of Corollary \ref{cor:utility_opt_special}} \label{app:utility_opt_special}
\begin{proof}
	The proof follows straightforwardly from Theorem~\ref{thm:utility_opt}. 
	A little thought gives the following: $k^*=1$ is an optimal solution if and only if $r(1) \leq 1$. 
	Solving $r(1) = \frac{2(\lambda+a)}{(2\lambda+a)} \cdot \frac{(n-1)\nu}{(n-1)\nu+a} \leq 1$ gives $\lambda \geq \frac{\nu(n-1)}{2} - \frac{a}{2}$.
	Similarly, $k^*=n$ is an optimal solution if and only if $r(n-1) \ge 1$. 
	Solving $r(n-1) = \frac{n((n-1)\lambda+a)}{(n-1)(n\lambda+a)} \cdot \frac{\nu}{\nu+a} \ge 1$ gives $\lambda \leq \frac{\nu}{n(n-1)} - \frac{a}{n}$.
\end{proof}

\subsection{Rewards $X_{k,t}$ Being i.i.d. over Time}\label{app:iid_explanation}
In this section, we provide an explanation for the following: the rewards, i.e., the AoI-based utility samples $X_{k,t}=U(\Delta(k,t))$, of each arm $k$ are \emph{i.i.d.} over time $t$, where $\Delta(k,t)$ is the AoI at the user's side when the user sends the $t$-th request and waits for the first $k$ responses. It suffices to argue that the AoI samples $\Delta(k,t)$ of each arm $k$ are \emph{i.i.d.} over rounds. Note that in the problem we consider, arm $k$ corresponds to waiting for the first $k$ responses; serving the $t$-th request corresponds to round $t$. In Eq.~\eqref{eq:aoik_l}, we have derived the expression of the AoI sample for arm $k$ in each round $t$ as follows (we remove the dependence on $t$ for ease of presentation):
\begin{equation}
\Delta(k) = R_{(k)} + \min_{i \in \mK} \Delta_i(s),
\end{equation}
where $R_{(k)}$ is the total waiting time for receiving the first $k$ responses, $\mK$ is the set of the indices of the servers that return the first $k$ responses, and $\min_{i \in \mK} \Delta_i(s)$ is the AoI of the freshest information among these $k$ responses at request time $s$.
We first argue that the first term $R_{(k)}$, which is the $k$-th smallest value among $n$ \emph{i.i.d.} exponential random variables, is \emph{i.i.d.} over rounds. This is true because the response time is exponentially distributed (with mean $1/\nu$) and is \emph{i.i.d.} across the servers. Second, we argue that the second term $\min_{i \in \mK} \Delta_i(s)$ is also \emph{i.i.d.} over rounds.
In the model we consider, the information updates for each server are assumed to follow a Poisson process, and thus, the inter-update time is exponentially distributed for each server and is \emph{i.i.d.} across the servers. Therefore, at any given request time $s$, the AoI at each server (i.e., the elapsed time since the last update) has the same distribution as the inter-update time due to the memoryless property of the exponential distribution. That is, random variable $ \Delta_i(s)$ is also exponentially distributed with mean $1/\lambda$ and is \emph{i.i.d.} across servers. Hence, random variable $\min_{i \in \mK} \Delta_i(s)$ is the minimum of $k$ \emph{i.i.d.} exponential random variables and is thus also \emph{i.i.d.} over rounds. Therefore, the AoI sample $\Delta(k)$ is \emph{i.i.d.} over rounds as it is the sum of two \emph{i.i.d.} random variables $R_{(k)}$ and $\min_{i \in \mK} \Delta_i(s)$.

\end{document}